\documentclass{vldb}

\usepackage{amsthm}
\usepackage{latexsym}

\usepackage{auto-pst-pdf}
\usepackage{footnote}
\usepackage{listings}
\usepackage{enumerate}

\usepackage{xcolor}
\usepackage{multirow}
\usepackage{graphicx}
\usepackage{url}

\usepackage[breaklinks]{hyperref}
\usepackage{breakurl}
\usepackage{xfrac}
\usepackage{balance}
\usepackage{srcltx}
\usepackage{xspace}
\usepackage{booktabs}
\usepackage{enumitem}
\usepackage{subcaption}
\usepackage{caption}

\usepackage{soul} 

\usepackage[ruled,vlined]{algorithm2e}
\usepackage{algpseudocode}
\usepackage{float}

\makeatletter
\def\@copyrightspace{\relax}
\makeatother

\def\endabstract{\if@twocolumn\else\endquotation\fi}

\usepackage{balance}

\graphicspath{{figs/}}

\newcommand{\newtext}[1]{#1}

\newcommand*{\refname}{Bibliography}

\newcommand\Mark[1]{\textsuperscript#1}

\newtheorem{theorem}{Theorem}[section]
\newtheorem{definition}[theorem]{Definition}
\newtheorem{conjecture}[theorem]{Conjecture}
\newtheorem{corollary}[theorem]{Corollary}

\vldbTitle{Chasing Similarity: Distribution-aware Aggregation Scheduling}
\vldbAuthors{Feilong Liu, Ario Salmasi, Spyros Blanas, Anastasios Sidiropoulos}
\vldbDOI{https://doi.org/TBD}
\vldbVolume{12}
\vldbNumber{xxx}
\vldbYear{2019}

\begin{document}

\title{Chasing Similarity: Distribution-aware \\Aggregation Scheduling
(Extended Version) \thanks{This is an extended version of the paper:
Feilong Liu, Ario Salmasi, Spyros Blanas, Anastasios Sidiropoulos.
``Chasing Similarity: Distribution-aware Aggregation Scheduling''.
\textit{PVLDB}, 12(3): 292-306, 2018 [25].}}

\numberofauthors{1}
\author{
	\alignauthor Feilong Liu\Mark{1}, Ario Salmasi\Mark{1}, Spyros Blanas\Mark{1}, Anastasios Sidiropoulos\Mark{2} \\
	\affaddr{\Mark{1}The Ohio State University, \Mark{2}University of Illinois at Chicago} \\
	\email{{\large \{}liu.3222,salmasi.1,blanas.2{\large \}}@osu.edu, sidiropo@gmail.com} \\
}

\maketitle

\nocite{grasp}


\begin{abstract}

Parallel aggregation is a ubiquitous operation in data analytics that is
expressed as \texttt{GROUP\;BY} in SQL, \texttt{reduce} in Hadoop,
or \texttt{segment} in TensorFlow.
Parallel aggregation starts with an optional local pre-aggregation
step and then repartitions the intermediate result across the network.
While local pre-aggregation works well for low-cardinality
aggregations, the network communication cost remains significant
for high-cardinality aggregations even after local pre-aggregation.
The problem is that the repartition-based algorithm for
high-cardinality aggregation does not fully utilize the network.

In this work, we first formulate a mathematical model that captures the
performance of parallel aggregation.
We prove that finding optimal aggregation plans from a known data
distribution is NP-hard, assuming the Small Set Expansion conjecture.
We propose GRASP, a GReedy Aggregation Scheduling Protocol
that decomposes parallel aggregation into phases.
GRASP is distribution-aware as it aggregates the most similar partitions
in each phase to reduce the transmitted data size in subsequent phases.
In addition, GRASP takes the available network bandwidth into account
when scheduling aggregations in each phase to maximize network utilization.
The experimental evaluation on real data shows that GRASP outperforms
repartition-based aggregation by $3.5\times$ and LOOM by $2.0\times$.

\end{abstract}


\vspace{-0.5em}
\section{Introduction}
\label{sec:intro}

\noindent
Aggregation is widely used in data analytics.
Parallel aggregation is executed in two
steps.
The first step is an optional local aggregation where data is
aggregated locally, followed by a second step where data is repartitioned and
transferred to the final destination node for 
aggregation~\cite{Shatdal:parallelagg95sigmod,graefe:93}.
The local aggregation can reduce the amount of data transferred
in the second step for algebraic aggregations, as tuples with the same
GROUP BY key are aggregated to a single tuple during local aggregation~\cite{Cieslewicz:2007vldb,ye:agg11damon,gagan:simd17,Polychroniou:2015simd,liwang:numaagg15}.
Local aggregation works effectively for low-cardinality domains, such as
\texttt{age}, \texttt{sex} or \texttt{country}, where data can be reduced
substantially and make the cost of the repartition step negligible.
However, high-cardinality aggregations see little or no benefit from local
aggregation.
Optimizing the repartitioning step for high-cardinality aggregations has
received less research attention.

High-cardinality aggregations are surprisingly common in practice.
One example is sessionization, where events in a timestamp-ordered log
need to be grouped into user sessions for analysis. 
An exemplar is the publicly-available Yelp dataset where
5.2M reviews are aggregated into 1.3M user sessions~\cite{yelpdataset}.
Even when there are no high-cardinality attributes, aggregation on
composite keys of multiple attributes can lead to
high-cardinality aggregations, which is common in data cube
calculations~\cite{datacube96icde}.

This paper focuses on reducing the communication cost for
high-cardinality aggregations.
We classify aggregations into two types: all-to-one aggregation and
all-to-all aggregation.
In all-to-one aggregation, one coordinator collects and aggregates data
from all compute nodes. 
All-to-one aggregation frequently happens at the last stage of a query.
In all-to-all aggregation, data is repartitioned on the GROUP BY
attributes and every node aggregates a portion of the data.
All-to-all aggregation is common in the intermediate stages of a query plan.

Directly transmitting the data to the destination node during an
aggregation underutilizes the network.
In all-to-one aggregation, 
the receiving link of the destination is the bottleneck while every
other receiving link in the network is idle.
In all-to-all aggregation, workload imbalance due to skew or non-uniform
networks~\cite{datacenter,luo18socc} means that some network links will be
underutilized when waiting for the slower or overburdened links to
complete the repartitioning.

\begin{figure*}[t]
\begin{minipage}{0.24\linewidth}
\centering
\includegraphics[]{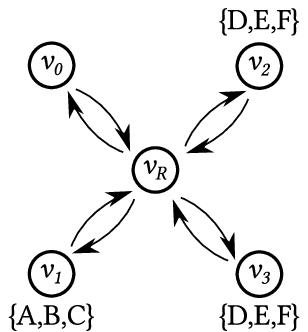}
\vspace{0.4em}
\caption{\small Graph representation of a cluster with four nodes. The
aggregation destination is $v_0$ and the router is $v_R$.}
\vspace{1.2em}
\label{fig:graphexample}
\end{minipage}
\hfill
\begin{minipage}{0.22\linewidth}
\vspace{-1.1em}
\centering
\includegraphics[]{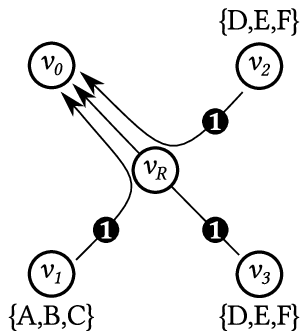}
\vspace{0.4em}
\caption{\small Aggregation based on \textbf{repartitioning} completes in 9 time
units. The bottleneck is the $v_R\!\to\!v_0$ link.}
\label{fig:graphexample-repart-phase1}
\end{minipage}
\hfill
\begin{minipage}{0.25\linewidth}
\vspace{-1.6em}
\centering
\includegraphics[scale=1.00]{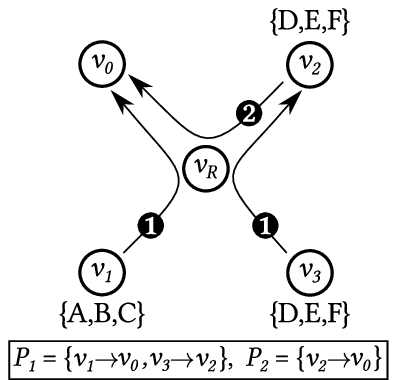}
\vspace{0.4em}
\caption{\small The \textbf{similarity-a\-ware} plan completes in 6 time units.}
\label{fig:graphexample-phase}
\end{minipage}
\hfill
\begin{minipage}{0.25\linewidth}
\vspace{-1.6em}
\centering
\includegraphics[scale=1.00]{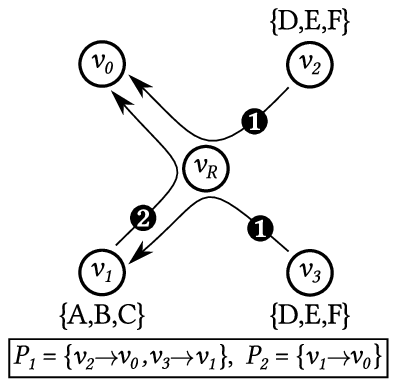}
\vspace{0.4em}
\caption{\small The \textbf{similarity-ob\-liv\-i\-ous} plan finishes in 9~time units.}
\label{fig:graphexample-phase-v2}
\end{minipage}
\vspace{-2em}
\end{figure*}


Systems such as Dremel~\cite{dremel10vldb},
Camdoop~\cite{camdoop12nsdi}, NetAgg~\cite{Mai14Netagg} and
SDIMS~\cite{sdims04sigcomm} reduce the communication cost and increase
network utilization by using \emph{aggregation trees} for all-to-one
aggregations.
The most relevant prior work is
LOOM~\cite{Culhane2014hotcloud,Culhane2015infocomm}, which 
builds aggregation trees \newtext{in a network-aware manner. 
LOOM assumes that every node stores $|R_{leaf}|$ distinct
keys and that the cardinality of the final aggregation result is
$|R_{root}|$.
Given these parameters as input, 
LOOM produces an aggregation tree with a fan-in that is a function of
the reduction rate $\sfrac{|R_{root}|}{|R_{leaf}|}$.  Applying LOOM
during query execution is not trivial, however, as the cardinality of
the input and the final result is not known in advance. 
(Even estimations of the cardinality can be
inaccurate \cite{leis18vldb}.)
Furthermore, the aggregation plan that LOOM produces fails to consider
how the \emph{similarity between partitions} impacts the
reduction rate at intermediate steps of the aggregation.
}

The importance of considering partition similarity 
during aggregation can be shown with an example.
Figure~\ref{fig:graphexample} shows an all-to-one aggregation in a
4-node cluster, where $v_R$ is the
switch, node $v_0$ is the destination node, node
$v_1$ stores three tuples with keys A, B and C, and nodes $v_2$ and $v_3$
store three tuples each with keys D, E and F. (For simplicity, the
figures only show the GROUP BY keys.)
\begin{itemize}[leftmargin=1em, topsep=0em, noitemsep]
\item
The \textbf{repartitioning} strategy in
Figure~\ref{fig:graphexample-repart-phase1} finishes the aggregation in
9 time units, 
where one time unit is the time $v_0$ needs to receive and process a single tuple.
\item
The \textbf{similarity-aware} aggregation plan
in Figure~\ref{fig:graphexample-phase} proceeds in two phases. 
In the first phase, $v_1$ transmits keys \{A,B,C\} to $v_0$ and $v_3$
transmits keys \{D,E,F\} to $v_2$. 
In the second phase, $v_2$ computes the partial aggregation and
transmits keys \{D,E,F\}.
The entire aggregation completes in 6 time units --- $1.5\times$
faster than repartitioning.
%
\item
The
\textbf{similarity-oblivious} aggregation plan shown in
Figure~\ref{fig:graphexample-phase-v2} transmits keys \{D,E,F\} from $v_3$
to $v_1$ in the first phase and then needs 6 time units in the second
phase to transmit keys \{A,B,C,D,E,F\} to $v_0$. The entire
aggregation completes in 9 time units, as fast as repartitioning.
\end{itemize}

\vspace{0.25em}
\noindent
This paper introduces GRASP, an algorithm that carefully constructs
aggregation plans to accelerate high-cardinality aggregation.
Unlike prior solutions~\cite{dremel10vldb, camdoop12nsdi, Mai14Netagg,
sdims04sigcomm} that do not consider if 
data can be combined during an aggregation,
GRASP 
\textit{aggregates fragments with similar keys first to improve performance.}
GRASP has the following attributes:
(1) it is distribution-aware as it 
selects aggregation pairs that will produce smaller partial aggregates,
(2) it is topology-aware
as it schedules larger data transfers on faster network links,
(3) it achieves high network utilization as it uses as many
network links as possible.

The paper is structured as follows.
%
Section~\ref{sec:problemdef} develops a theoretical model for the network
cost of parallel data aggregation.
Section~\ref{sec:graspalg} introduces GRASP, a topology-aware and data
distribution-aware algorithm, that accelerates
aggregations by leveraging partition similarity. 
A natural question to ask is if GRASP produces aggregation plans that 
approximate the optimal plan by some constant factor.
Section~\ref{sec:hardness} proves that the aggregation
scheduling problem 
cannot be approximated
within a constant factor by any polynomial algorithm (including GRASP),
assuming the SSE conjecture.
Section~\ref{sec:exp} contains the experimental evaluation 
which shows that GRASP 
can be up to $3.5\times$ faster than repartitioning and up to
$2.0\times$ faster than LOOM on real datasets.

\vspace{-0.5em}
\section{Problem definition}
\label{sec:problemdef}

\begin{table}[b]
\caption{Symbol definitions.}
\label{tab:problemdef}
	\small
\centering
\begin{tabular}{cl}
\toprule
\textbf{Symbol}                    &\multicolumn{1}{c}{\textbf{Description}}                                                                                                  \\ 
\midrule
$s \rightarrow t $    & Data transfer from node $s$ to node $t$                   \\ 
$P_i$                       & Phase $i$, $P_i=\{s_1\!\rightarrow\!t_1,\;s_2\!\rightarrow\!t_2,\;\ldots\}$ \\ 
$\mathbb{P}$              & Aggregation plan, $\mathbb{P} = \{P_1, P_2, \ldots\}$\\
$X_i^l\left(v\right )$      & Data of partition $l$ in node $v$ after $P_i$ completes \\ 
$X_i\left(v\right )$      & Data in $v$ after $P_i$ finishes, $X_i\left(v \right) = \bigcup_{l} X_i^l(v)$\\ 
$X_0\left(v\right )$      & Data in $v$ before the aggregation starts \\ 
$Y_i\left(s \rightarrow t\right )$      & Data sent from $s$ to $t$ in phase $P_i$ \\ 
$w$      & Size of one tuple \\ 
$B\left(s \rightarrow t \right )$                  & Available bandwidth for  the $s \rightarrow t$ data transfer \\ 
$\texttt{COST}\left(s\!\rightarrow\!t \right)$ & Network cost for the $s \rightarrow t$ data transfer \\ 
\bottomrule
\end{tabular}
\end{table}

\begin{figure*}[t]
\centering
\includegraphics[width=\textwidth]{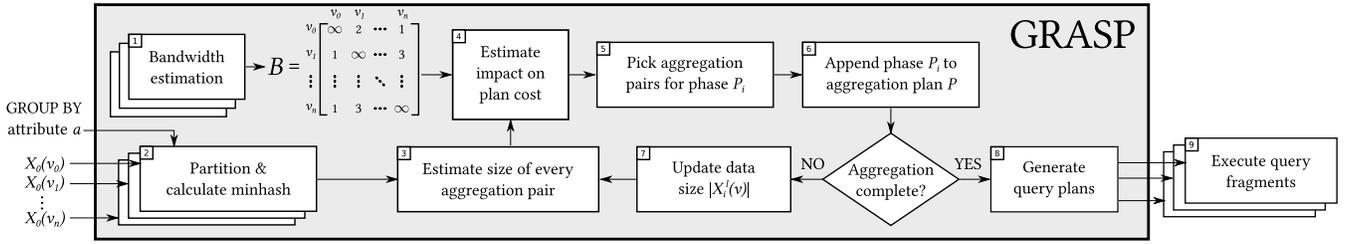}
\vspace{0.4em}
\caption{The GRASP framework.}
\vspace{-1.0em}
\label{fig:systeoverview}
\end{figure*}

\noindent
We use a connected, directed, weighted graph $G =(V(G),$
$E(G))$ to represent the network topology of the cluster.
Each edge $\langle v_i, v_j \rangle \in E\left(G\right )$ represents one
network link, with the edge direction to be the direction of data flow.

The fat-tree topology is widely used
in data centers~\cite{fattree}.
We represent all routers in the network as a single node $v_R \in
V\left(G\right )$ and model the fat-tree topology as a star network.
The set $V_C = V\left(G\right ) - \left \{v_R \right \}$ represents
the compute nodes of the cluster.
Compute nodes have bidirectional network links, therefore
$E\left(G\right ) = \left \{\langle s, v_R \rangle | s \in V_C\right \}$ 
$\bigcup_{\;}^{\;}\left \{\langle v_R, t \rangle | t \in V_C\right \}$,
where edge $\langle s, v_R \rangle$ represents the uplink and
edge $\langle v_R, t \rangle$ represents the downlink.

\subsection{Modeling all-to-one aggregations}
\label{sec:modelallto1}
\noindent
\textbf{Aggregation Model.}
We first consider an aggregation where data is aggregated to
one single node $v^* \in V_C$.
The aggregation consists of multiple phases which execute in serial order.
We use $\mathbb{P}$ to denote an aggregation execution plan with $n$ phases, 
$\mathbb{P} = \left \{P_1, P_2, ..., P_n\right \}$,
where $P_i$ represents one phase of the aggregation.
In a phase $P_i$, there are $k$ concurrent data transfers,
$P_i=\left \{s_1\to t_1, ..., s_k\to t_k\right \}$, where
$s_j \to t_j$ denotes the data transfer in which node $s_j$
sends all its data to node $t_j$.
Figure~\ref{fig:graphexample-phase} shows an aggregation execution plan
$\mathbb{P}$ with two phases $P_1$ and $P_2$.
Phase $P_1$ performs two data transfers
$v_1 \to v_0, v_3 \to v_2$, and phase $P_2$
performs one data transfer $v_2 \to v_0$.

We impose one constraint in the selection of $s \to t$ pairs:
node $s$ will never send its data to a node $t$ that has no data, unless
$t$ is the final destination node $v^*$,
as no data will be aggregated in this case.
(In fact, we could not find any instance where transferring to an empty
node $t$ would be beneficial over transmitting data directly to the
destination $v^*$ in a single-path star topology.)
Hence, a node can be a receiver multiple times across multiple phases, 
but once it transmits its data in some phase $P_i$
it becomes inactive and
it will not participate in the aggregation in phases $P_{i+1}, ..., P_{n}$. 
A corollary is that a node cannot be both sending and receiving data in
the same phase.

Let 
$X_0\left(v\right )$ 
be the data in node $v$ in the beginning of the aggregation execution
and
$X_i\left(v\right )$
be the data in node $v$ after phase $P_i$ completes. 
Let $Y_i(s\to t)$ be the data sent from $s$ to $t$ in phase $P_i$.
A node will send all its local data within one phase, hence  
$Y_i(s\to t) = X_{i-1}(s)$.
After phase $P_i$ completes, 
for every transfer $s \to t \in P_i$,
$X_i\left(s\right ) = \varnothing$ and 
\begin{equation}
\small
X_i\left(t\right ) = 
X_{i-1}\left(t\right ) \bigcup_{\;}^{\;}
\left(\bigcup_{s \to t  \in P_i }^{\;}X_{i-1} \left(s\right ) \right )
\label{eq:updatedata:all-to-one}
\end{equation}

The aggregation has finished in phase $n$ when all nodes except $v^*$
have sent their data out for aggregation:
\begin{equation}
\forall v \in \big \{ V_C - \left \{ v^* \right \} \big \} : X_n\left(v\right ) = \varnothing
\label{eq:stopcondition}
\end{equation}

\vspace{0.5em}
\noindent
\textbf{Aggregation cost.}
The aggregation execution plan $\mathbb{P} = \left \{P_1, ..., P_n\right \}$ consists of phases in serial order. Hence the network cost of $\mathbb{P}$ is:
\begin{equation}
\texttt{COST}\left(\mathbb{P}\right ) = \sum \texttt{COST}\left(P_i\right )
\label{eq:totalcost}
\end{equation}
The network cost 
for phase $P_i = \left \{s_1\to t_1, ..., s_k\to t_k\right \}$ 
is the cost of the network transfer which completes last:
\begin{equation}
\texttt{COST}\left(P_i\right ) = \underset{ s_j \to t_j \in P_i }{max} \texttt{COST}\left(s_j \to t_j\right ) 
\label{eq:phasecost}
\end{equation}

\noindent
The cost of the data transfer $s_j\to t_j$ is the time it takes to
transfer $|Y_{i}(s_j\!\to\!t_j)|$ tuples of size $w$ each over
the available link bandwidth $B\left(s_j\!\to\!t_j\right )$:
\begin{equation}
\texttt{COST}\left(s_j\to t_j\right ) = \frac{|Y_{i}(s_j \to t_j)| \cdot w}{B\left(s_j \to t_j\right )}
\label{eq:datatransfercost}
\end{equation}
Section~\ref{sec:alg:estbw} shows how GRASP estimates
$B\left(s_j \to t_j\right )$ without network topology information.
Section~\ref{subsec:bandwidth} shows one way to calculate
$B\left(s_j \to t_j\right )$ if all network activity is known.

\vspace{0.5em}
\noindent
\textbf{Problem definition.}
Given a connected, directed, weighted graph $G$, 
the data $X_0\left(v\right )$ in every node $v \in V_C$, 
the final destination node $v^* \in V_C$,
obtain an aggregation execution plan containing one or more phases $\mathbb{P} = \left \{P_1, P_2, ..., P_n\right \}$ such that
$\texttt{COST}\left(\mathbb{P}\right )$ is minimized.

\subsection{Modeling all-to-all aggregations}
\noindent
The all-to-all aggregation model executes multiple all-to-one
aggregations over different partitions in a single plan. 

\vspace{0.5em}
\noindent
\textbf{Aggregation Model.}
In all-to-all aggregation 
data is divided into $m$ partitions, $L = \left \{l_1, l_2, ..., l_m\right \}$.
Every compute node in $V_C$ is the aggregation destination for one or
more partitions.
This is specified by a mapping $M: L \to V_C$ that maps a
partition $l \in L$ to a specific destination $v \in V_C$.
Let $X_0^l\left(v\right )$ be the
data of partition $l$ in node $v$ in the beginning of the aggregation
execution and 
$X_i^l\left(v\right )$ be the data of partition
$l$ in node $v$ after phase $P_i$ completes. 

Within one aggregation phase, 
a node $s$ will send an entire partition $l$ of local data to $t$, 
hence $Y_i(s\to t) = X_{i-1}^{l}(s) \subseteq X_{i-1}(s)$.
Once a node transmits all its data for partition $l$ it becomes
inactive in subsequent phases for this partition,
but it will participate in aggregations for other active partitions.
Hence, in all-to-all aggregation a node can be both sending and
receiving data in the same phase, as long as it does not send and
receive data belonging to the same partition.
$X_i\left(v\right )$ 
is the data in node $v$ after phase $P_i$ completes:
\begin{equation}
\small
X_i\left(v\right ) = 
X_{i-1}\left(v\right ) \bigcup_{\;}^{\;}
\left(\bigcup_{s \to v  \in P_i }^{\;}Y_{i} \left(s \to v\right ) \right ) - 
\bigcup_{v\to t \in P_i}^{\;}Y_{i}(v\to t)
\label{eq:updatedata:all-to-all}
\end{equation}
All-to-all aggregation completes when data in all partitions are aggregated to
their corresponding destination:
\begin{equation}
\forall \; l\to v^* \in M : 
\forall v \in \big \{ V_C - \left \{ v^* \right \} \big \} : 
X_n^l\left(v\right ) = \varnothing
\label{eq:stopconditionalltoall}
\end{equation}

\vspace{0.5em}
\noindent
\textbf{Problem definition.}
Given a connected, directed, weighted graph $G$,
the data $X_0^l\left(v\right )$ for each partition $l \in L$ in every
node $v \in V_C$,
and a mapping $M: L \to V_C$ denoting the destination of each partition,
obtain an aggregation execution plan containing one or more phases $\mathbb{P} = \left \{P_1, P_2, ..., P_n\right \}$ such that
$\texttt{COST}\left(\mathbb{P}\right )$ is minimized.


\section{The GRASP framework}
\label{sec:graspalg}

\noindent
This section introduces GRASP, a \underline{gr}eedy \underline{a}ggregation 
\underline{s}ched\-ul\-ing \underline{p}rotocol,
which uses partition similarity as a heuristic to carefully schedule data 
transfers to improve performance.



\subsection{Overview}

\noindent
Figure~\ref{fig:systeoverview} shows an overview of the GRASP framework.
The inputs to the framework
are the data $X_0(v)$ in every node $v$ 
and the \textsc{Group By} attribute $a$. 
The input data may be either a table in the database or
an intermediate result produced during query processing.
Steps~\ref{step:bwest},~\ref{step:minhash}
and~\ref{step:queryexec}
are run by all compute nodes, while steps
\ref{step:intersect}--\ref{step:genfrag} are run in the coordinator.

\begin{enumerate}[label=\textbf{\arabic*}), wide=0pt, ref={\arabic*}, nosep]
\item \textbf{Bandwidth estimation.} \label{step:bwest}
Every node estimates the available 
bandwidth between itself and other nodes and stores it in matrix $B$.
Section~\ref{sec:alg:estbw} describes the process in detail.
\item \textbf{Partition, pre-aggregate and calculate minhash signatures.} \label{step:minhash}
Every node partitions and aggregates data locally.
During this operation, every node runs 
the minhash algorithm~\cite{jaccard98,
lsh98,lsh99} to produce succinct
minhash signatures.

\item \textbf{Estimate the cardinality of every possible pair.} \label{step:intersect}
The coordinator collects the minhash signatures and estimates the
cardinality of all possible aggregation pairs.
An aggregation pair is a partition $l$, a source node $s$ and a
destination node $t$.
Section~\ref{sec:algo:minhash} presents the algorithms in detail.

\item \textbf{Estimate the cost of the final plan.} \label{step:heuristic}
The coordinator uses the available bandwidth matrix $B$ as input and 
estimates the runtime cost and the future benefit of executing every possible
aggregation pair. 
Section~\ref{sec:alg:heuristic} describes the cost heuristic.

\item \textbf{Generate aggregation phase $P_i$.} \label{step:aggselect}
The coordinator
selects aggregation pairs for phase $P_i$ based on their cost. 
The detailed algorithm is described in Section~\ref{sec:alg:aggselect}.


\item \textbf{Add $P_i$ to aggregation plan $\mathbb{P}$.} \label{step:addtoplan}
If the aggregation is complete, the aggregation plan $\mathbb{P}$ is scheduled for
execution.

\item \textbf{Update data size $|X_{i}^{l}(v)|$.} \label{step:updatesize}
The coordinator updates the estimation of the size
of each partition $|X_{i}^{l}(v)|$ in every node for the next phase of
the aggregation.
GRASP does not make another pass over the data,
as the minhash signature of any intermediate result 
can be calculated from the original minhash signatures obtained in
Step~\ref{step:minhash}.

\item \textbf{Generate query plans.} \label{step:genfrag}
The aggregation planning is complete.
GRASP generates query plans for execution.

\item \textbf{Query execution.} \label{step:queryexec}
Every node in the cluster 
executes its assigned aggregations for each phase.
\end{enumerate}


\subsection{Estimating the bandwidth}
\label{sec:alg:estbw}


\noindent
This section describes how GRASP estimates the available
bandwidth for data transfers without network topology information.
GRASP schedules aggregation plans so that one node sends to and receives from at most one
node within a phase to avoid network contention.
This ensures that the outgoing link
and the incoming link of each node are used by at most one data transfer.
Similar approaches are used by R{\"{o}}diger et al.~\cite{wolf15vldb} to
minimize network contention.



GRASP measures the pair-wise bandwidth through a benchmarking
procedure that is executed on system startup.
The bandwidth $B(s \rightarrow t)$ is measured by running a benchmark on
every $s$ and $t$ pair individually, where $s$ keeps sending data to $t$. 
The average throughput is stored as the estimation of $B(s \rightarrow
t)$ in a matrix, where the row index is the sender and the column index
is the receiver.
(For example, $B(v_0 \rightarrow v_1)=2$ in Figure~\ref{fig:systeoverview}.)
The bandwidth matrix $B$ is computed once and reused for all queries that follow.
Section~\ref{sec:exp:bwest} evaluates the accuracy of the estimation and
the robustness of GRASP to estimation errors.


\subsection{Estimating the size of intermediate results}
\label{sec:algo:minhash}

\noindent
GRASP needs to estimate the cardinality of the intermediate result
between every node pair $s$ and $t$ for aggregation planning.
According to set theory, the size of the union of two sets $S$ and $T$ 
can be calculated
as $|S \cup T| = |S| + |T| - |S \cap T| = \frac{|S|+|T|}{1+J}$, 
where $J$ is the Jaccard similarity
$J = \frac{|S\cap T|}{|S \cup T|}$.
Hence one can calculate the cardinality of an aggregation from the
cardinality of the input partitions $S$, $T$ and the Jaccard similarity between them.

Accurately calculating the Jaccard similarity is as expensive as
computing the aggregation itself, as it requires collecting both
inputs to the same node.
GRASP thus estimates the Jaccard similarity using the minhash
algorithm~\cite{jaccard98, lsh98,lsh99}.
After running minhash, the inputs are represented 
by a small vector of integers called a \emph{minhash signature}.
The minhash signatures are used to estimate the Jaccard similarity between
the two sets.

The minhash algorithm generates minhash signatures by applying
a set of hash functions to the dataset.
The minhash signature value is the minimum value produced by each hash
function.
Figure~\ref{fig:minhashexample} shows an example of the minhash
signature calculation for two sets $S$ and $T$ and their minhash
signatures $sig(S)$ and $sig(T)$, respectively.
The Jaccard similarity between the two sets can be estimated from the
minhash signatures as
the fraction of the hash functions which produce
the same minhash value for both sets.
In the example shown in Figure~\ref{fig:minhashexample},
the accurate Jaccard similarity is 
$J_{acc}=\frac{|S\cap T|}{|S \cup T|}=\frac{6}{10}$. 
The estimated Jaccard similarity from the minhash signatures is
$J_{est}=\sfrac{1}{2}$, as only hash function $h_2(\cdot)$ produces
the same minhash value between the two sets.

\begin{figure}[t]
\centering
\includegraphics[width=\columnwidth]{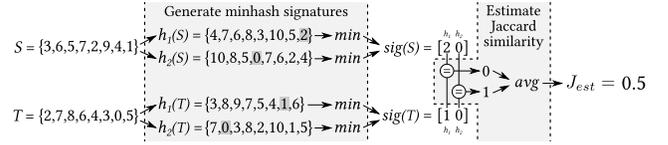}
\vspace{0.2em}
\caption{
Example of Jaccard similarity estimation with the minhash algorithm and
hash functions $h_1(x)=(x+1)\; mod\; 11$ and $h_2(x)=(3x+1)\; mod\; 11$.
}
\label{fig:minhashexample}
\vspace{-1em}
\end{figure}

Another appealing property of the minhash algorithm is that the
minhash signature $sig(S \cup T)$ can be computed from the minhash
signatures $sig(S)$ and $sig(T)$, respectively: The minhash signature of
the union is the pairwise minimum of the respective signatures, or
$sig(S \cup T)[i] = min\big(sig(S)[i],$ $sig(T)[i]\big)$.
The practical significance of this property is that GRASP needs to access
the original data only once before the aggregation starts, and then
will operate on the much smaller signatures during aggregation planning.

In GRASP, every node 
partitions the local data and 
calculates the cardinality and the minhash signatures for each partition. 
(This is step~\ref{step:minhash} in Figure~\ref{fig:systeoverview}.)
The coordinator collects the cardinality and the minhash signature for
each partition of every node in two arrays \texttt{Card} and
\texttt{MinH} of size $|V_C| \times |L|$.
The arrays are initialized to
\texttt{Card}$[v,l] \leftarrow \big|X_0^l(v)\big|$ and 
\texttt{MinH}$[v,l] \leftarrow sig\big(X_0^l(v)\big)$. 
After these arrays are populated with information from every node, 
they are only accessed by two functions during aggregation planning,
which are defined in Algorithm~\ref{alg:minhashfunc}.
The first function is \textsc{EstCard}$(s,t,l)$
which estimates the Jaccard similarity between the sets $X_i^l(s)$ and
$X_i^l(t)$ from their minhash signatures and returns an estimate of the
cardinality of their union.
The second function is \textsc{Update}$(s,t,l)$
which updates the \texttt{Card} and \texttt{MinH} arrays
after the $s \to t$ transfer of partition $l$. 

How many hash functions does minhash need?
GRASP uses only 100 hash functions so that
signatures are less than 1KB.
This choice sacrifices accuracy but keeps the computation and network cost small.
Satuluri and Partha\-sa\-rathy~\cite{Satuluri:2012} show that the
estimation is within 10\% of the accurate similarity with $95\%$ probability
when $n=100$.
Section~\ref{sec:exp:realdata} evaluates the accuracy 
of the minhash estimation.

{
\vspace{-1em}
\begin{algorithm}[h]

\caption{\textsc{EstCard}$(s,\!t,\!l)$ estimates $\left |X_i^l(s) \cup
X_i^l(t)\right |$ and \textsc{Update}$(s,\!t,\!l)$ updates the \texttt{Card}
and \texttt{MinH} arrays.}

\label{alg:minhashfunc}
\small

\SetKwProg{myfunc}{function}{}{}

\textbf{Input}
$s, t \in V_C$: computing node identifiers \\
\hspace{2.8em}
$l \in L$: data partition identifier \\

\myfunc{\textsc{EstCard}$(s,\!t,\!l)$}
{

\nl $\texttt{sigS} \leftarrow \texttt{MinH}[s,l]$;
 $\texttt{sigT} \leftarrow \texttt{MinH}[t,l]$;
 $J \leftarrow 0$ \\
\nl \For{$j \in [1,n]$}
{
\nl \If{\emph{\texttt{sigS}}$[j] =$ \emph{\texttt{sigT}}$[j]$}
	{
		\nl $J \leftarrow J + \sfrac{1}{n}$
	}
}
\nl \Return $\frac
{\texttt{Card}[s,l]\;+\;\texttt{Card}[t,l]}
{1\;+\;J}$  \\

}

\myfunc{\textsc{Update}$(s,\!t,\!l)$}
{

\nl $\texttt{Card}[t,l] \leftarrow \textsc{EstCard}(s,\!t,\!l)$ \\

\nl $\texttt{Card}[s,l] \leftarrow 0$ \\

\nl \For{$j \in [1,n]$}
{
		\nl $\texttt{MinH}[t,l][j] \leftarrow min(\texttt{MinH}[s,l][j],\texttt{MinH}[t,l][j])$ \\ 
		
}

\nl $\texttt{MinH}[s,l] \leftarrow \;\perp$ \\ 

}
\end{algorithm}
}


\subsection{Forecasting the benefit of each aggregation}
\label{sec:alg:heuristic}

\noindent
Ideally one should take the cost of all future aggregation phases
into account when picking the best plan for the current phase.
This is prohibitively expensive as there are $n^{n-2}$ possible
aggregation trees for a cluster with $n$ nodes \cite{cayley1889}.
A greedy approach that minimizes the cost of the current phase only ignores how
similarity can reduce the network cost of future data transfers.
Hence, GRASP looks one phase ahead during optimization to balance the
network transfer cost of a data transfer in the current phase with the
anticipated future savings from transmitting less data in the next phase.

The heuristic GRASP uses to pick which transfers to schedule in the
current phase is based on 
a cost function $C_i(s, t, l)$ that adds the cost of an $s \to
t$ transfer in this phase and the cost of transmitting the union of the
data in the next phase.
$C_i(s, t, l)$ is constructed based on the following intuition:

\vspace{0.5em}
\noindent \textbf{1)} 
Penalize the following transfers 
by setting $C_i = \infty$ so that they will never be picked:
(1) Node $s$ sending partitions whose destination is $s$,
to prevent circular transmissions.
(2) One node sending a partition to itself, 
as this is equivalent to a no-op.
(3) Transfers involving nodes that neither have any data
nor are they the final destination for this partition. 

\vspace{0.5em}
\noindent \textbf{2)}
When any node transmits partition $l$ to its final destination $M(l)$, only
the cost of the data transfer needs to be considered, as this partition
will not be re-transmitted again.
Hence, we set $C_i$ to $\texttt{COST}(s\rightarrow t)$ in this case, 
where $\texttt{COST}$ is defined in Eq.~\ref{eq:datatransfercost}, 
and $Y_i\left ( s \rightarrow t \right ) = X_{i-1}^{l}\left ( s \right )$.

\vspace{0.5em}
\noindent \textbf{3)} 
Otherwise, add the cost of the $s \rightarrow t$ transfer to the 
cost of transmitting the aggregation result in the next phase.
We define
$E_i(s, t, l) = \frac{\textsc{EstCard}(s,t,l) \cdot w}{B(s\rightarrow t)}$
to simplify the notation.
\vspace{0.5em}

Based on the above, we define $C_i$ for a transfer $s \to t$ of
partition $l$ between any pair of nodes $(s, t)$ in phase $P_i$ as:
\begin{equation}
C_i(s, t, l)=\left\{\begin{matrix}
\infty  & s=t \\
\infty  & s=M(l)  \\
\infty  & X_{i-1}^{l}\left ( s \right ) = \varnothing \\
\infty  & X_{i-1}^{l}\left ( t \right ) = \varnothing \\
\texttt{COST}(s\rightarrow t) & t = M(l) \\ 
\texttt{COST}(s\rightarrow t) + E_i(s, t, l) & otherwise  \vspace{0.5em}\\
\end{matrix}\right.
\label{eq:cost}
\end{equation}

Figure~\ref{fig:costexample} shows $C_1$ for the
phase $P_1$ of the aggregation shown in Figure~\ref{fig:graphexample}.
There is only one partition in this example, hence $l=0$.
The row index is the sending node and the column index is the receiving node.
Note that the matrix $C_i$ will not be symmetric, because transfers
$s \rightarrow t$ and $t \rightarrow s$ transmit different data and
use different network links.

%
%

%
%

\begin{figure}[t]
\centering
\includegraphics[]{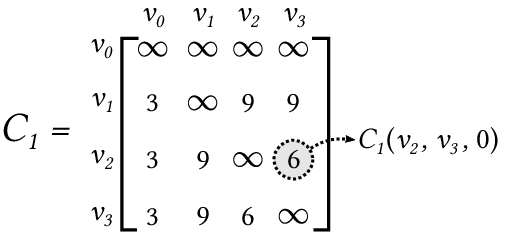}
\vspace{1.0em}
\caption{
\small
The matrix $C_1$ for the
phase $P_1$ of the aggregation problem in
Figure~\ref{fig:graphexample} that has a single partition. 
The example assumes $w$ is equal to the bandwidth $B$.
Rows represent the sender and columns represent the receiver.
The circled value corresponds to the aggregation $v_2 \rightarrow v_3$
where $v_2$ sends out $\left \{ D,E,F \right \}$ to aggregate with
$\left \{ D,E,F \right \}$ in $v_3$.
}
\label{fig:costexample}
\end{figure}


\newcommand{\grinf}{\color{lightgray}\infty}
\newcommand{\vare}{\varepsilon}
\newcommand{\nsp}{\hspace{-0.75em}}

\subsection{Selecting aggregation pairs}
\label{sec:alg:aggselect}

{
\begin{algorithm}[b]
\small
\caption{Selecting data transfers for phase $P_i$}
\label{alg:aggselect}
\textbf{Input}
$C_i$: the cost function defined in Eq.~\ref{eq:cost}.
\\
\hspace{2.7em}
$V_{send}$: candidate nodes to send 
\\
\hspace{2.7em}
$V_{recv}$: candidate nodes to receive
\\
\hspace{2.7em}
$V_{l}$: candidate nodes to operate on partition $l$
\\
\textbf{Output}
$P_i$: the next aggregation phase\\

\nl 
$P_i \leftarrow \varnothing$; 
$V_{send} \leftarrow V_C$; 
$V_{recv}\leftarrow V_C$;
$V_l \leftarrow V_C$

\nl \While{$|V_{send}| > 0$ and $|V_{recv}| > 0$} {\label{line:terminate}

\nl Pick $\langle s \rightarrow t, l \rangle$ such that
\\
\hspace{1.5em}$s \in \left(V_{send} \cap V_l \right)$, 
$t \in \left(V_{recv} \cap V_l \right)$ and
\\
\hspace{1.5em}$C_i(s, t, l)$ has the minimum value in $C_i$\label{line:pickedge}

\nl\If{$C_i(s, t, l) = \infty$}
{
	\nl break\label{line:break}
}

\nl Remove $s$ from $V_{send}$ and $V_{l}$, if found \\
\nl Remove $t$ from $V_{recv}$ and $V_{l}$, if found \label{line:removenode}

\nl Add $\langle s \rightarrow t, l \rangle$ to $P_i$ \label{line:addedge}

\nl \textsc{Update}$(s,t,l)$ \label{line:updatestat}
}

\nl \Return $P_i$

\end{algorithm}
}

\noindent
This section describes step~\ref{step:aggselect} in
Figure~\ref{fig:systeoverview} which selects transfers among 
all possible pairs to
produce one aggregation phase $P_i$.
There are three aspects for consideration when
selecting candidate aggregations:

\vspace{0.5em}
\noindent \textbf{1)} 
\textbf{In each phase, how many transfers does a node participate in?}
Prior work shows that uncoordinated network communication leads
to congestion in the network~\cite{wolf15vldb, Rodiger:14icde}.
R{\"{o}}\-di\-ger et al.~\cite{wolf15vldb} do application-level
scheduling by dividing communication into stages to improve throughput,
where in each stage a server has a single target to send to and a single
source to receive from.
Like prior work, GRASP restricts the communication within one phase 
to minimize network contention.
Specifically, GRASP picks transfers such that
one node sends to at most one node and receives from at most one
node in each aggregation phase.

\vspace{0.5em}
\noindent \textbf{2)} 
\textbf{How many nodes are selected for aggregation in one phase?}
In order to maximize the network utilization, GRASP picks as many data
transfers as possible in one phase until the available bandwidth $B$ is
depleted.

\vspace{0.5em}
\noindent \textbf{3)} 
\textbf{Given many candidate aggregation pairs,
which aggregation should one choose within one phase?}
GRASP minimizes the $C_i$ function
defined in Equation~\ref{eq:cost} and
selects aggregations by 
picking the smallest $C_i$ values.
\vspace{0.5em}

Algorithm~\ref{alg:aggselect} shows how GRASP
selects candidate aggregations for one phase $P_i$.
$V_{send}$ is the set of candidate nodes to be senders, 
$V_{recv}$ is the set of candidate nodes to be receivers, 
and $V_l$ is the nodes that can operate on partition $l$.
The algorithm picks the aggregation pair
which has smallest value in $C_i$ (line~\ref{line:pickedge}).
The algorithm then removes the selected nodes from the candidate node sets
(lines~6-\ref{line:removenode}) to enforce that
(a) one node only sends to or receives from at most one node, and 
(b) one node does not send and receive data for the same partition
within the same phase.
Then, the transfer $s \to t$ for partition $l$ is added to the
aggregation phase $P_i$ (line~\ref{line:addedge}).
GRASP calls the function \textsc{Update}$(s,t,l)$, which was defined in
Algorithm~\ref{alg:minhashfunc}, to update
the minhash signatures and the cardinalities 
in arrays \texttt{MinH} and \texttt{Card} (line~\ref{line:updatestat}),
as data in $s$ and $t$ will change after the aggregation.
The algorithm stops when either candidate set is empty
(line~\ref{line:terminate}) or there are no more viable transfers in this
phase (line~\ref{line:break}).

\begin{figure}[t]
\centering
\includegraphics[scale=0.9]{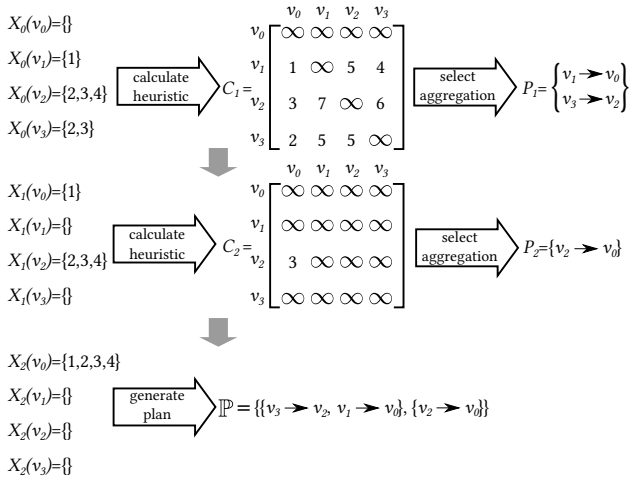}
\vspace{1em}
\caption{An example of how GRASP generates aggregation plans for an all-to-one aggregation
with a single partition.}
\label{fig:algoexample}
\vspace{-1em}
\end{figure}

Figure~\ref{fig:algoexample} shows an example of how GRASP selects
aggregations using the $C_i$ cost function.
For simplicity,
we again show an all-to-one aggregation with a single partition $l=0$, and 
we assume the bandwidth $B$ to be equal to the tuple width $w$.
In the first iteration, the coordinator constructs the matrix
$C_1$ from the cost function described
in Section~\ref{sec:alg:heuristic}.
For example, assume in the first phase $|X_0(v_2)| = 3$ and $|X_0(v_2)
\cup X_0(v_3)| = 3$, then $C_1(v_2, v_3, 0) = 6$.
After constructing the cost matrix $C_1$, 
GRASP picks data transfers for aggregation using Algorithm~\ref{alg:aggselect}.
The first pick is $v_1 \rightarrow v_0$ because it has the least cost.
Because a transfer has now been scheduled on the $v_1 \rightarrow v_0$
link, GRASP eliminates $v_1$ and $v_0$ from the
corresponding candidate sets.
GRASP then picks $v_3 \rightarrow v_2$. 
GRASP then finishes this phase because
there are no candidates left, and appends the aggregation phase 
$P_1 = \left \{ v_1 \rightarrow v_0, v_3 \rightarrow v_2\right \}$ to
the aggregation plan $\mathbb{P}$.
In the next iteration, GRASP constructs matrix $C_2$ and picks the
last data transfer $v_2 \to v_0$ for phase $P_2$. 
At this point, all data will have been aggregated to the
destination nodes so the aggregation plan $\mathbb{P}$ will be scheduled
for execution.


\newtheorem{remark}{Remark}[section]

\section{Hardness of Approximation}\label{sec:hardness}

\noindent
Many hard problems are amenable to efficient approximation algorithms
that quickly find solutions that are within a guaranteed distance to the
optimal.
For instance, $2$-ap\-prox\-i\-ma\-tion algorithms ---polynomial algorithms
that return a solution whose cost is at most twice the optimal--- are
known for many NP-hard minimization problems.
A natural question to ask is how closely does GRASP approximate the 
optimal solution to the aggregation problem. 

This section proves that it is not
feasible to create a polynomial algorithm that approximates the optimal
solution to the aggregation problem within any constant factor.
In other words, the aggregation problem is not only NP-hard
but it also cannot be approximated within any constant factor by any
polynomial algorithm, including GRASP.
This hardness of approximation result is much stronger than
simply proving the NP-hardness of the problem, as many NP-hard problems
are practically solvable using approximation.%

The proof is structured as follows.
Section \ref{subsec:bandwidth} introduces an assumption
regarding the cost of using shared network links.
Section \ref{subsec:SSE} defines the Small Set
Expansion (SSE) problem and the well-established SSE conjecture. 
Section \ref{subsec:NP-one} starts with an instance of SSE and reduces
it to the all-to-one aggregation problem. 
This proves that the all-to-one aggregation problem is NP-hard to
approximate, assuming the SSE conjecture.
Section \ref{subsec:NP-all} proves that the all-to-all
aggregation problem is also NP-hard to approximate.

\subsection{Link sharing assumption}\label{subsec:bandwidth}

\noindent
Whereas GRASP 
will never schedule concurrent data transfers on the same link in one
phase in a star network, the theoretical proof needs a mechanism to
assess the runtime cost of sharing a network link for multiple transfers.
Our proof makes the fair assumption that the cost of sending data from one
node to another is proportional to the total data volume that
is transferred over the same link across all aggregations in this phase.

One way to incorporate link sharing information in the cost calculation
is to account for the number of concurrent data transfers on the $s \to t$ 
path when computing the available bandwidth $B(s \to t)$.
For example, for the network topology shown in
Figure~\ref{fig:graphexample} the available bandwidth from $s$ to
$t$, $B\left( s \rightarrow t \right)$ can be calculated as:

\begin{equation}
B\left(s \rightarrow t\right ) = min\left(\frac{W\left(\langle s, v_R \rangle\right )}{d_o\left(s\right )}, \frac{W\left(\langle v_R, t \rangle\right )}{d_i\left(t\right )} \right )
\label{eq:linkbandwidth}
\end{equation}

\noindent
where $W\left(\langle s, v_R\rangle\right )$ and $W\left(\langle v_R,
t\rangle\right )$ are the network bandwidths of the links,
$d_o\left(s\right )$
denotes the number of data transfers using the $\langle s, v_R \rangle$ link and
$d_i\left(t\right )$
denotes the number of data transfers using the $\langle v_R, t \rangle$
link in this phase.

\subsection{The Small Set Expansion Problem}\label{subsec:SSE}
\noindent
This subsection defines the Small Set Expansion (SSE) conjecture
\cite{raghavendra2010graph}. We first briefly discuss the 
intuition behind this problem 
then give a formal definition.

\subsubsection{Intuition}
\noindent
A $d$-regular graph is a graph where each vertex has $d$ edges for some
integer $d \geq 1$. 
The Small Set Expansion problem asks if there exists a small subset of
vertices that can be easily disconnected from the rest in a 
$d$-regular graph. 
The SSE conjecture states that it is NP-hard to distinguish between the
following two cases: 
(1) The \textbf{YES} case, there exists
some small set of vertices that can be disconnected from the graph. 
(2) The \textbf{NO} case, such a set does not exist. In other words,
in this case
every set of vertices has a relatively large boundary to the other
vertices in the graph.

Note that the SSE conjecture is currently open, as it has not been
proven or disproven yet. 
Just like the well-known $\mathbf{P} \neq \mathbf{NP}$ conjecture,
the theory community has proceeded to show that many problems are hard to
approximate based on the general belief that the SSE conjecture is true.
Significant hardness of approximation results that assume the SSE
conjecture include the treewidth and
pathwidth of a graph \cite{austrin2012inapproximability}, 
the Minimum Linear Arrangement (MLA) and the $c$-Balanced Separator
problem \cite{raghavendra2012reductions}.

\subsubsection{Formal Definition}
\noindent
Let $G$ be an undirected $d$-regular graph.
For any subset of vertices $S \subseteq V(G)$, we define the \emph{edge
expansion} of $S$ to be $\Phi(S) = \frac{E(S,V\!\setminus\! S)}{d|S|}$.

\begin{definition}
Let $\rho \in [-1,1]$. Let $\Phi^{-1}$ be the inverse function of the
normal distribution. Let $X$ and $Y$ be jointly normal random variables
with mean $0$ and covariance matrix 
$\begin{pmatrix} 
   1 & \hspace{-0.75em} \rho\\ 
\rho & \hspace{-0.75em} 1 
\end{pmatrix}$. 
We define $\Gamma_\rho \colon [0,1] \to [0,1]$ as
$
\Gamma_\rho(\mu) = \Pr[X \leq \Phi^{-1}(\mu) \land Y \leq \Phi^{-1}(\mu)]
$.
\end{definition}


\begin{conjecture}[The Small Set Expansion conjecture \cite{raghavendra2010graph}]\label{conj:sse}
For every integer $q > 0$ and $\varepsilon, \gamma > 0$, it is NP-hard
to distinguish between the following two cases:
\begin{description}
\item{\textbf{YES}} There is a partition of $V(G)$ into $q$ equi-sized
sets $S_1,\ldots,S_q$ such that 
$\Phi(S_i) \leq 2\varepsilon$, $\forall i \in \{1,\ldots,q\}$.

\item{\textbf{NO}} For every $S \subseteq V(G)$ we have $\Phi(S) \geq 1
- (\Gamma_{1-\varepsilon/2}(\mu) + \gamma)/\mu$, where $\mu =
|S|/|V(G)|$.
\end{description}
\end{conjecture}


\begin{remark}\label{rem:yes}
In the \textbf{YES} case, the total number of edges that are not
contained in one of the $S_i$ sets is at most $2 \varepsilon |E|$.
\end{remark}

\begin{remark}\label{rem:no}
In the \textbf{NO} case, for every $S \subseteq V(G)$ with $|V(G)|/10
\leq |S| \leq 9|V(G)|/10$, we have $|E(S,V(G)\setminus S)| \geq c
\sqrt{\varepsilon}|E(G)|$, for some constant $c > 0$.
\end{remark}

\subsection{Hardness of the aggregation problem}\label{subsec:NP-one}
\noindent
Before stating the formal inapproximability result, we first provide the 
intuition behind our proof strategy approach. We then reduce the SSE
problem to the all-to-one aggregation problem. Finally, we show that
the all-to-all problem is a straightforward generalization of the all-to-one problem.

\subsubsection{Intuition}
\noindent
We now give a brief intuitive overview of the proof. Recall that in the
SSE problem we are given a graph $G$ and the goal is to decide whether
$G$ admits a partition into small subgraphs, each having a small boundary
(a SSE partition henceforth), or $G$ is an expander at small scales,
that is, all small subgraphs of G have a large boundary. The SSE
conjecture asserts that this problem is hard to approximate, and has
been used to show the inapproximability of various graph optimization
problems \cite{austrin2012inapproximability}. Inspired by these results,
we show that the all-to-one aggregation problem is hard to approximate
by reducing the SSE problem to it. Our proof strategy is as follows. We
begin with an instance $G'$ of the SSE problem. We encode $G'$ as an
instance of the all-to-one aggregation problem by interpreting each node
of $G'$ as a leaf node in the star network, and each edge
$\langle u,v \rangle$ of $G$ as a data item which is replicated in nodes $u$ and $v$
in the aggregation problem. We show that any partition of $G$ can be
turned into an aggregation protocol, and, conversely, any aggregation
protocol can be turned into a partition of $G$. The key intuition is
that the cost of the partition is related to the cost
of the aggregation via the observation that the data items that need to
be transmitted twice are exactly the edges that are cut by the
partition.

\subsubsection{Formal proof for the all-to-one aggregation}
\noindent
Suppose that we are given an all-to-one aggregation instance: a graph
$G$, a single destination vertex $v^* \in V(G)$, and the data $X_0(v)$
in each node $v \in V(G)$. 
Let $X = \bigcup_{v \in V(G)} X_0(v)$ be the set of all data.
Let $\mathbb{P} = \{P_1,P_2,\ldots,P_n\}$ be an execution
plan. For every $P_i = \{s_1 \rightarrow t_1, \ldots, s_k \rightarrow
t_k\} \in \mathbb{P}$, let $S(P_i) = \{s_1,\ldots,s_k\}$ and $T(P_i) =
\{t_1,\ldots,t_k\}$. 

We define the \emph{overhead cost} of $\mathbb{P}$
to be $\texttt{COST}\left(\mathbb{P}\right ) - |X|$. 
Under the all-to-one aggregation model, 
every execution plan is obtained from an aggregation tree.
To simplify the proof, we assume that one node sends
data to only one node within a phase. 
This modeling assumption is acceptable from a theoretical standpoint as
one can represent a phase where a node transmits data to multiple
destinations as a sequence of discrete phases to each individual destination.
We say that $\mathbb{P}$ is obtained from an \emph{aggregation tree}
$T_{P}$, if the following conditions hold: 

\begin{enumerate}[leftmargin=1.5em, itemsep=0em, topsep=0em]
\item
$T_P$ is a spanning tree of $G$, rooted at $v^*$. 
\item
The leaf vertices of $T_P$ are exactly the elements of $S(P_1)$. 
Furthermore, for every $i \in \{2,\ldots,k-1\}$, the leaf vertices of
$T_P \setminus \bigcup_{1 \leq j<i} S(P_j)$ are exactly the elements of
$S(P_i)$. 
\end{enumerate}


\begin{theorem}\label{thm:main}
For every $\varepsilon > 0$, given an aggregation instance $\big(G,
v^* \in V(G), X_0(v) \; \forall v \in V(G) \big)$, it
is SSE-hard to distinguish between the following two cases:
\begin{description}[itemsep=0em]

\item{\textbf{YES}} There exists an execution plan that is
obtained from an aggregation tree with overhead cost $O(\varepsilon
|X|)$.
\item{\textbf{NO}} Every execution plan that is obtained
from an aggregation tree has overhead cost $\Omega(\sqrt{\varepsilon}
|X|)$.

\end{description}
\end{theorem}

\begin{proof}
We start with an instance of SSE with $q = 1/\varepsilon$, and reduce it
to our problem. Let $G'$ be the $d$-regular graph of the SSE instance.
We construct an aggregation instance as follows.
Let $V(G) = V(G')$, and $X = E(G')$. Note that $G$ is a complete graph
with the same vertex set as $G'$. For every $v \in V(G)$, let $X_0(v) =
\{\langle u,w\rangle \in X: v=u \lor v=w\}$ be the set of data that is held by $v$.

In the \textbf{YES} case of the SSE instance, we have disjoint sets
$S_1,S_2,\ldots,S_q$ of equal size. For every $i \in \{1,2,\ldots,q\}$,
we have $|S_i| = |V(G)|/q = \varepsilon |V(G)|$. We may assume w.l.o.g.~that $v^*
\in S_q$. For every $i \in \{1,2,\ldots,q-1\}$, pick an arbitrary vertex
$v_i \in S_i$. Let also $v_q = v^*$. For every $j \in \{1,2,\ldots,q\}$,
let $\{s_{i,1},\ldots,s_{i,i_j}\} = S_j \setminus \{v_j\}$.
We first construct an aggregation tree $T$ as follows. For every $i \in
\{1,2,\ldots,q\}$, let $v_i$ be the parent of all other vertices in
$S_i$. Let $v_q$ be also the parent of $v_1,v_2,\ldots,v_{q-1}$.

Now consider the execution plan corresponding to $T$. This
aggregation has two phases: $\mathbb{P} = \{P_1,P_2\}$. First
we describe $P_1$. 
For each $S_i$, we aggregate all the data held by vertices of $S_i$ to
$v_i$; that is every vertex in $S_i$ (except $v_i$ itself) transfers its
dataset to $v_i$. This can be done simultaneously for all $S_i$'s, since
$S_i$'s are disjoint sets. We have that $P_1 = \{s_{1,1}\!\rightarrow\!v_1,\; 
s_{1,2}\!\rightarrow\!v_1,\; ...,\; s_{1,i_1}\!\rightarrow\!v_1,\; 
s_{1,2}\!\rightarrow\!v_2,\; ...,\; s_{1,i_2}\!\rightarrow\!v_2,\; 
...,\; 
s_{q,1}\!\rightarrow\!v_q,\; ...,\; s_{q,i_q}\!\rightarrow\!v_q \}$.


By the construction, at the beginning for each vertex $v$ we have that
$|X_0(v)| = d$. Therefore, for every $S_i$, the total volume of data to
be transferred to $v_i$ is $2\varepsilon |E(G)|=d\varepsilon |V(G)|=d|S_i|$.
In other words, for every $(s_{i,j} \rightarrow v_i) \in P_1$, we have
that $\texttt{COST}(s_{i,j} \rightarrow v_i) = 2\varepsilon |E(G)|$, and
thus we have $\texttt{COST}(P_1) = 2\varepsilon |E(G)|$.

In the second phase of the execution plan, for every $i \in
\{1,2,\ldots,q-1\}$, we need to transfer all the data held by $v_i$ to
$v^*$. This can be done simply by sending one data at a time to $v^*$.
We have:
\[
P_2 = \{v_1 \rightarrow v_q, v_2 \rightarrow v_q ,\ldots, v_{q-1} \rightarrow v_q\}
\]

By Remark \ref{rem:yes}, the total number of tuples that are
transferred more than once in this phase is at most $\varepsilon d |V(G)|
=2\varepsilon|E(G)|$. This means that $\texttt{COST}(P_2) \leq
(1+2\varepsilon) |E(G)|$. Therefore we have that
$\texttt{COST}(\mathbb{P}) \leq (1+4\varepsilon)|E(G)|$, and thus the
overhead cost of this execution plan is $O(\varepsilon|E(G)|)$.

In the \textbf{NO} case, we want to show that every execution
plan that is obtained from an aggregation tree has cost
$\Omega(\sqrt{\varepsilon}|E|)$. Let $\mathbb{P}$ be an 
execution plan that is obtained from an aggregation tree $T$. For every
$v \in V(T)$, let $T_v$ be the subtree of $T$ rooted at $v$.

Suppose that $v^*$ has a child $v$ such that $|V(T)|/10 \leq |V(T_v)|
\leq 9|V(T)|/10$. 
We apply Remark \ref{rem:no} by setting $S = T_v$. 
We have that $E(S,V(G)\setminus S) \geq c \sqrt{\varepsilon}|E(G)|$, for
some constant $c > 0$. 
This means that there are at least $c \sqrt{\varepsilon}|E(G)|$ data
that are going to be sent at least twice to $v^*$ in the execution plan,
or $\texttt{COST}(\mathbb{P}) = \Omega((1
+\sqrt{\varepsilon})|E(G)|)$. 
Thus, the overhead cost of this execution plan is
$\Omega(\sqrt{\varepsilon}|E(G)|)$.

Otherwise, $v^*$ has a child $v$ such that $|V(T_v)| < |V(T)|/10$. 
In this case, there are at least $9|E(G)|/10$ data in $T_v$ that are going 
to be transferred at least twice to get to $v^*$ in the execution plan.
Therefore, we have $\texttt{COST}(\mathbb{P}) = \Omega((0.9 +
0.9)|E(G)|)$, and thus the overhead cost of this execution plan is
clearly $\Omega(\sqrt{\varepsilon}|E(G)|)$. This completes the proof.
\end{proof}

\vspace*{-0.75em}
\begin{corollary}
Assuming Conjecture \ref{conj:sse}, it is NP-hard to approximate the
minimum overhead cost of an all-to-one aggregation plan that
is obtained from an aggregation tree within any constant factor.
\end{corollary}

\vspace*{-0.75em}
\begin{corollary}\label{cor:mainresult}
Assuming Conjecture \ref{conj:sse}, it is NP-hard to find an all-to-one
aggregation plan that is obtained from an aggregation tree
with minimum cost.
\end{corollary}

One might ask if it is feasible to brute-force the problem for small
graphs by enumerating all possible aggregation trees and picking the
best solution. 
Unfortunately this would be extremely expensive even for small graphs. 
Cayley's formula \cite{cayley1889} states that the number of different
spanning trees of graph with $n$ vertices is $n^{n-2}$.
Hence, even for $n=20$ one needs to enumerate $20^{18} \geq 10^{23}$
different trees.

%

\subsubsection{Formal proof for the all-to-all aggregation}\label{subsec:NP-all}
\noindent
The more general case is the all-to-all aggregation problem. 
We observe that the all-to-one aggregation problem can be trivially reduced to the
all-to-all aggregation problem,
since by the definition, every instance of the all-to-one aggregation
problem is also an instance of the all-to-all aggregation problem.

\begin{theorem}\label{thm:main_result}
Assuming Conjecture \ref{conj:sse}, it is NP-hard
to find an all-to-all aggregation plan with minimum cost.
\end{theorem}

\begin{proof}
We reduce the all-to-one aggregation problem to the all-to-all
aggregation problem. Suppose that we are given an instance of the
all-to-one aggregation problem. By its definition, this is
also an instance of the all-to-all aggregation problem where the mapping
$M$ is such that the aggregation destination of every partition is 
node $v^* \in V_C$.
By Corollary
\ref{cor:mainresult} we know that the all-to-one aggregation problem is
NP-hard assuming Conjecture \ref{conj:sse}, therefore the all-to-all
aggregation problem is NP-hard as well.
\end{proof}


\section{Experimental evaluation}
\label{sec:exp}
\noindent
This section compares the GRASP algorithm with 
the repartitioning algorithms and LOOM.
Section~\ref{sec:exp:setup} introduces the experimental setup, which includes the hardware setting, the workloads and the baselines.
The other sections evaluate the following questions:
\begin{itemize}[leftmargin=*, itemsep=0.15em]
\item 
($\S$~\ref{sec:exp:similarity})
How well does GRASP leverage similarity between datasets?

\item
($\S$~\ref{sec:exp:overlap})
How does similarity within the dataset affect performance?

\item 
($\S$~\ref{sec:exp:loadimbalance})
Can GRASP benefit from workload imbalance?

\item 
($\S$~\ref{sec:exp:bwest})
How accurate is the bandwidth estimation? How robust is GRASP to estimation errors?

\item 
($\S$~\ref{sec:exp:nonuniformnetwork})
How does GRASP perform in nonuniform networks?

\item
($\S$~\ref{sec:exp:scaleout})
How does the performance change when the number of fragments increases?

\item 
($\S$~\ref{sec:exp:realdata})
Is GRASP faster than aggregation based on repartitioning and LOOM on TPC-H and
real datasets?

\item 
($\S$~\ref{sec:exp:ec2})
How well does GRASP work in a real-world deployment where the
network conditions are unpredictable? 

\end{itemize}

\vspace{0.5em}
\subsection{Experimental setup}
\label{sec:exp:setup}

\noindent
We implemented the GRASP framework in C++ and we have open-sourced our
prototype implementation~\cite{graspcode}.
We evaluate GRASP in two clusters.
The first is a shared cluster connected by a 1 Gbps network.
Each machine has two NUMA nodes with two Intel Xeon
E5-2680v4 14-core processors and 512 GB of memory.
The second cluster is Amazon EC2 with d2.8xlarge instances
which have 36 vCPUs and 244 GB of memory. 
The instances are connected with a 10 Gbps network.

We run one or more aggregation fragments in each machine/instance. 
Hence, one
fragment corresponds to one logical graph node in Figure~\ref{fig:graphexample}.
\newtext{We evaluate all-to-all aggregations by setting the mapping between
partitions and destinations so that aggregation results are evenly balanced
across all nodes.
We evaluate all-to-one aggregations by mapping all data
partitions to the same destination.}

Our evaluation reports the total response time to complete the
aggregation query. All our performance results include the time to plan the
aggregation using GRASP, the time to transfer all data to their destinations 
and the time to process the aggregation locally in each node.
All experiments use hash-based local aggregations.

\vspace{0.5em}
\subsubsection{Baselines}


\noindent
We compare GRASP with two baselines.
\newtext{The first baseline is LOOM~\cite{Culhane2014hotcloud,Culhane2015infocomm}.
As described in Section~\ref{sec:intro}, LOOM needs the size of aggregation results
during query planning.}
In our evaluation we configure LOOM to use the accurate result size so that LOOM achieves its
best performance.
The second baseline is repartitioning which has two versions.
One version is without local aggregation, where data
is directly sent to the destination fragment for aggregation.
We use ``Repart'' to denote this version.
The other version is with local aggregation, where data is first aggregated locally,
then the local aggregation result is sent to the destination fragment for aggregation.
We use ``Preagg+Repart'' to denote this version of repartitioning.
Note that repartitioning works for both all-to-all and all-to-one aggregations, while 
LOOM only works for all-to-one aggregations.

\subsubsection{Workloads}
\noindent
We use five workloads in our evaluation.

\textbf{1) Synthetic workload.} The first workload is a synthetic workload which has one table \texttt{R}, with two long
integers \texttt{R.a} and \texttt{R.b} as attributes.
The query evaluated is 
\texttt{SELECT R.a SUM(R.b) FROM R GROUP BY R.a}.

\textbf{2) TPC-H workload.} The second workload is the TPC-H workload with scale factor 80.
We evaluate this subquery from TPC-H Q18:
\texttt{SELECT ORDERKEY, SUM(QUANTITY) FROM LINEITEM GROUP BY ORDERKEY}.
%
The \texttt{LINEITEM} table is partitioned and distributed on the
\texttt{SUPPKEY} to framgents with a modulo hash function.

\textbf{3) MODIS workload.} The third workload is the Surface Reflectance data MOD09 from MODIS (Moderate
Resolution Image Spectroradiometer)~\cite{modis09}.
The MODIS data provides the surface relfectance of 16 bands together with the location
coordinates (latitude and longitude).
In the processing of MODIS data, one product is MOD09A1~\cite{modis09A1} which aggregates the observed
data in an 8-day period
with the following query:
\texttt{SELECT Latitude, Longitude, MIN(Band3) FROM RelfectTable GROUP BY ROUND(Latitude, 2),
\\ROUND(Longitude, 2) WHERE Date BETWEEN `01/01/2017' AND `01/08/2017'}.
%
The MODIS data is stored in separate files, one file per satelite image in
timestamp order.
We download about 1200 files from the MODIS website, and assigned files into
plan fragments in a round-robin fashion.
Overall, there are about 3 billion tuples and 648 million distinct GROUP BY
keys in this dataset.

\textbf{4) Amazon workload.} The fourth dataset is the Amazon review dataset~\cite{amazondataset}.
The review dataset has more than 82 million reviews from about 21 million users.
The dataset includes the reviewer ID, overall rating, review time and detail review etc.
We evaluate the following query to calculate the average rating a customer gives out.
\texttt{SELECT ReviewerID, AVG(OverallRate) FROM AmazonReview GROUP BY ReviewerID}.
%
The reviews are stored in timestamp order and we split this file into plan fragments.

\textbf{5) Yelp workload.} The fifth dataset is the Yelp review dataset~\cite{yelpdataset}.
The review dataset has more than 5 million reviews from about 1.3 million users.
The Yelp dataset has similar attributes as the Amazon dataset and we use a similar query
to calculate the average stars a customer gives.

\subsection{Experiments with uniform bandwidth}
\label{sec:uniformnetwork}

\noindent
This section evaluates GRASP in a setting where each
plan fragment communicates with the same bandwidth.
The measured inter-fragment bandwidth is 118 MB/s.
We experiment with 8 machines and 1 fragment per machine, which results
in 8 fragments in total.
We use the synthetic workload in this section.

\subsubsection{Effect of similarity across fragments}
\label{sec:exp:similarity}

\noindent
GRASP takes advantage of the similarities between datasets
in different fragments in aggregation scheduling.
How well does the GRASP algorithm take advantage of
similarities between datasets?

\begin{figure}[t]
\centering
\begin{subfigure}[b]{0.23\textwidth}
\centering
\includegraphics[scale=1.0]{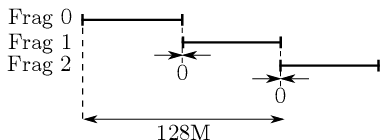}
\caption{Jaccard similarity $J=\frac{0}{128}$.}
\label{fig:similarity_dataset:0}
\end{subfigure}
\vspace{0.3em}
\begin{subfigure}[b]{0.23\textwidth}
\centering
\includegraphics[scale=1.0]{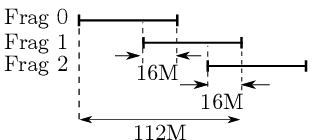}
\caption{Jaccard similarity $J=\frac{16}{112}$.}
\label{fig:similarity_dataset:1}
\end{subfigure}
\vspace{0.3em}
\caption{
The line segments represent the range of GROUP BY attributes.
The Jaccard similarity increases when the overlap of GROUP BY key ranges increases.
}
\vspace{-1.2em}
\label{fig:similarity_dataset}
\end{figure}

In this experiment, we
change the similarities between datasets, i.e. the number
of common GROUP BY keys, in different plan fragments.
Each plan fragment has 64 million tuples.
Figure~\ref{fig:similarity_dataset} shows how we change the similarity between datasets.
Each segment in Figure~\ref{fig:similarity_dataset} shows the range of $R.a$ in one 
fragment. Figure~\ref{fig:similarity_dataset} only shows fragments 0, 1 and 2.
The range of datasets between adjacent fragments has an overlap.
The Jaccard similarity increases when the size of the overlap increases.

The experimental results for all-to-one aggregation are shown in Figure~\ref{fig:exp:similarity}.
The horizontal axis is the Jaccard similarity coefficient between datasets.
The vertical axis is the speedup over the Preagg+Repart algorithm with Jaccard similarity 0.
Here speedup 1 corresponds to response time of 64.6 seconds.
Figure~\ref{fig:exp:similarity} shows that GRASP has the best performance and
is up to $4.1\times$ faster than Preagg+Repart and $2.2\times$ faster than LOOM when the Jaccard similarity is 1.
Figure~\ref{fig:exp:similarity} shows that the performance of Repart and Preagg+Repart stays
the same when the Jaccard similarity changes.
This means that repartitioning cannot utilize the similarities between datasets.

GRASP has better performance than LOOM for two reasons.
First, GRASP is data distribution-aware and prioritizes aggregations with higher similarity.
Second, GRASP has higher network utilization than LOOM.
In GRASP, a fragment can be both sending and receiving as long as it is not working on the same
partition.
In LOOM, a fragment is either a parent fragment receiving data or a child fragment sending data.

In all-to-all aggregation GRASP has similar performance with repartitioning as
there is no underutilized link in the network. We omit the results for brevity.

\begin{figure*}[t]
\begin{minipage}{0.32\linewidth}
\centering
\includegraphics[]{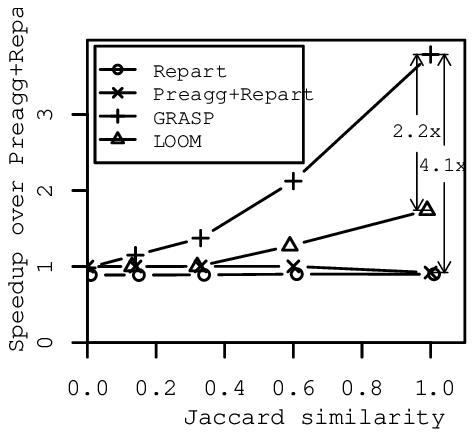}
\vspace{0.4em}
\caption{Speedup of GRASP when the similarity between datasets increases. GRASP is up to $2.2\times$ faster than LOOM and $4.1\times$ faster than Preagg+Repart.}
\label{fig:exp:similarity}
\end{minipage}
\hfill
\begin{minipage}{0.33\linewidth}
\centering
\includegraphics[scale=0.9]{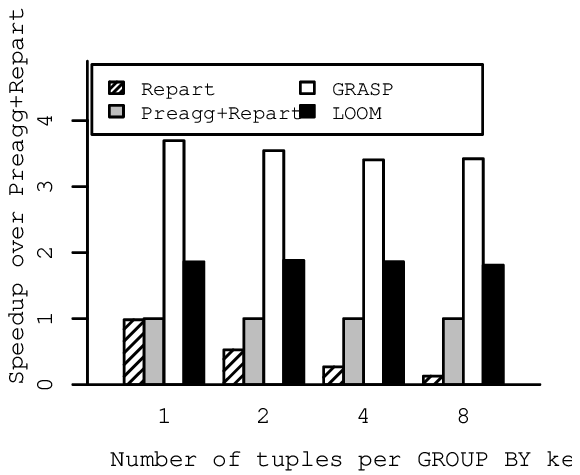}
\vspace{0.4em}
\caption{Speedup over Preagg+Repart when there are multiple tuples for each GROUP BY key in the same fragment for all-to-one aggregation.}
\label{fig:exp:localagg}
\end{minipage}
\hfill
\begin{minipage}{0.31\linewidth}
\centering
\includegraphics[]{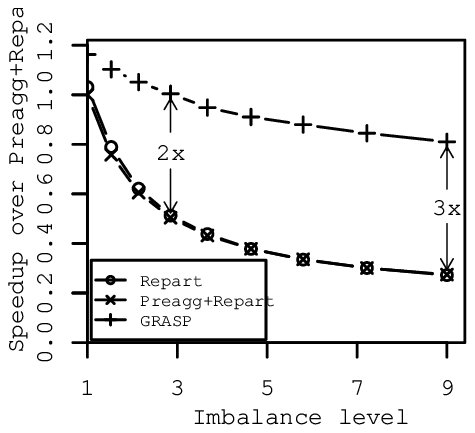}
\vspace{0.4em}
\caption{Speedup of GRASP for all-to-all aggregations when there fragment 0 receives more tuples.
GRASP is up to $3\times$ faster than Preagg+Repart.}
\label{fig:exp:workloadimb}
\end{minipage}
\end{figure*}

\subsubsection{Effect of similarity within fragments}
\label{sec:exp:overlap}

\noindent
This experiment evaluates how GRASP works when there are multiple
tuples for one GROUP BY key within one fragment.
In this case 
local aggregation will 
reduce the size of data, hence the Preagg+Repart algorithm will have better
performance than the Repart algorithm.

There are 128 million tuples in each fragment in this experiment.
We change the distinct cardinalities
of the datasets from 128 million, 64 million, 32 million to 16 million,
which changes the number of tuples per GROUP BY key from 1, 2, 4, to 8, respectively.
The smaller the distinct cardinality is, the more tuples are aggregated during local aggregation.

The results for the all-to-one aggregation are shown in Figure~\ref{fig:exp:localagg}.
The horizontal axis is the number of tuples for each GROUP BY key within the same fragment.
The vertical axis shows the speedup over the Preagg+Repart algorithm.
Higher bars means better performance.
The results show that Preagg+Repart has better performance than Repart when the number of tuples
for each GROUP BY key increases, which means there are more opportunities for local
aggregation.
However, GRASP always has better performance: it is more than $3\times$ faster than
Preagg+Repart and about $2\times$ faster than
than LOOM in all-to-one aggregations.
Hence GRASP has the same or better performance than repartition and LOOM when the similarity
within the same dataset changes.


\vspace{-0.2em}
\subsubsection{Effect of workload imbalance}
\label{sec:exp:loadimbalance}

\noindent
In parallel aggregation, some fragments may receive more tuples to aggregate
for two reasons.
First, the repartition function may assign more GROUP BY keys to some fragments.
Second, even if each fragment gets the same number of GROUP BY keys to process,
there may be skew in the dataset.
In this section, we evaluate
how GRASP works when one fragment gets more tuples to process.

In this experiment, we have 128 million tuples and $R.a$
ranges from 1 to 128 million.
We change the repartition function to assign more tuples to fragment 0.
We assign $n$ million tuples to fragment 0 for aggregation
and assign $m = \frac{128-n}{7}$
million tuples to the other fragments.
We use $l = \frac{n}{m}$ to denote the \textit{imbalance level}.
When $n$ equals to 16, $l$ is 1 and there is no imbalance.
However, as $n$ increases, fragment 0 gets more tuples than other fragments.

The results are shown in Figure~\ref{fig:exp:workloadimb}.
The horizontal axis is \textit{imbalance level} 
$l$.
The vertical axis is the speedup over Preagg+Repart when $l$ is 0.
Here speedup 1 corresponds to response time of 22.1 seconds.
Notice that LOOM is not shown here because LOOM does not work for all-to-all aggregations.
Figure~\ref{fig:exp:workloadimb} shows that the performance of repartition and GRASP both
decreases
when the workload imbalance increases.
However, the performance decreases much faster for repartition than GRASP and GRASP is
already $2\times$ faster than Preagg+Repart when fragment 0 receives about 3 times
of data of other fragments.
This is because in repartition, other fragments will stop receiving and aggregating data when they are waiting for fragment 0 to complete.
While for GRASP, other fragments are still scheduled to receive and aggregate data.
GRASP improves performance when some fragments process
more tuples.


\subsection{Experiments with nonuniform bandwidth}
\label{sec:nonuniformnetwork}

\noindent
GRASP is cognizant of the network topology, which is crucial 
when the communication
bandwidth is nonuniform. Non\-uniform bandwidth means that some plan fragments communicate at
different speeds than others.
The distribution of the link bandwidth is not uniform in 
many common network topologies.
Datacenter networks are often oversubscribed
and data transfers within the same rack
will be faster than data transfers across racks~\cite{datacenter}.
The data transfer throughput between instances in the
cloud is also nonuniform~\cite{luo18socc}.
Even HPC systems which strive for balanced networks may have nonuniform
configurations~\cite{hpctopo}.

This section evaluates how GRASP performs when the network bandwidth is
nonuniform.
All experiments in this section run multiple concurrent plan
fragments in each server to emulate a nonuniform network 
where some data transfers will be faster than others due to locality.

\begin{figure}[b!]
\vspace{-2em}
\begin{minipage}{0.42\linewidth}
\includegraphics[scale=1.0]{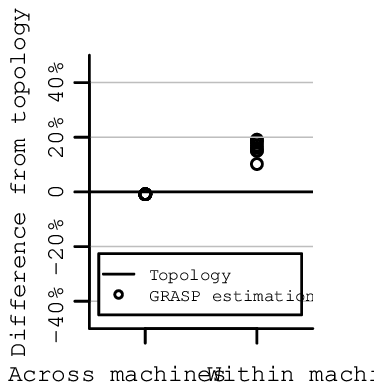}
\caption{Comparing between the theoretical bandwidth and the bandwidth
estimated from
	benchmarks.}
\label{fig:bwcomp}
\end{minipage}
\hspace{1em}
\begin{minipage}{0.53\linewidth}
\includegraphics[scale=1.0]{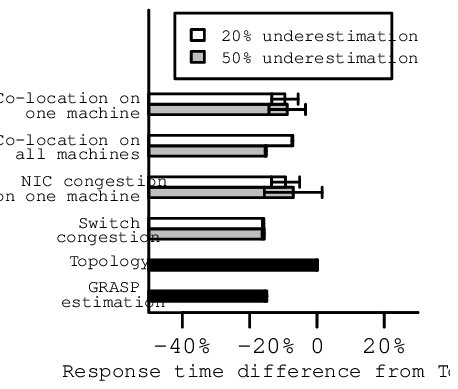}
\caption{Speedup on the MODIS dataset when changing the estimated bandwidth.}
\label{fig:bwest}
\end{minipage}
\end{figure}

\begin{figure*}[t]
\begin{minipage}{0.21\linewidth}
\includegraphics[scale=0.8]{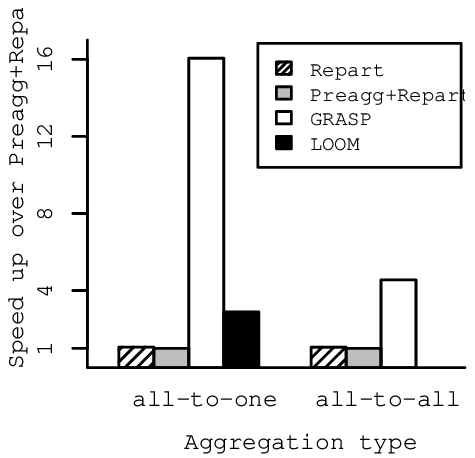}
\vspace{0.2em}
\caption{Speedup over Preagg+Repart with nonuniform bandwidth.}
\label{fig:exp:topology}
\end{minipage}
\hfill
\begin{minipage}{0.47\linewidth}
\centering
\includegraphics[]{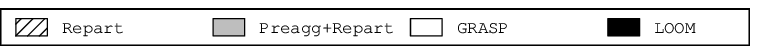}
\subcaptionbox{All-to-one aggregation.\label{fig:exp:scaleout:alltoone}}{
\includegraphics[scale=0.8]{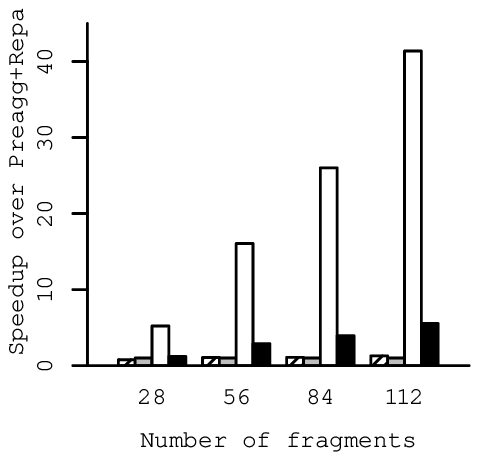}
}
\subcaptionbox{All-to-all aggregation.\label{fig:exp:scaleout:alltoall}}{
\includegraphics[scale=0.8]{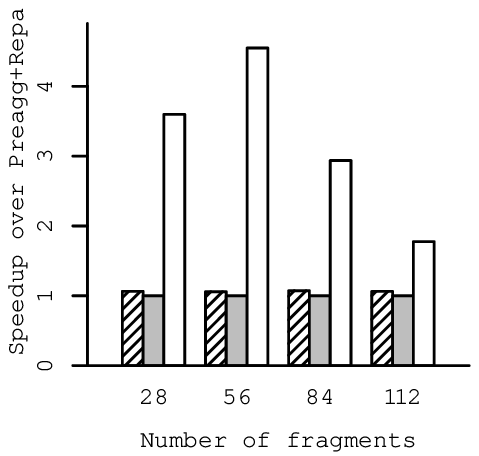}
}
\vspace{0.2em}
\caption{Speedup over Preagg+Repart when scaling out.}
\label{fig:exp:scaleout}
\end{minipage}
\hfill
\begin{minipage}{0.28\linewidth}
\includegraphics[scale=0.8]{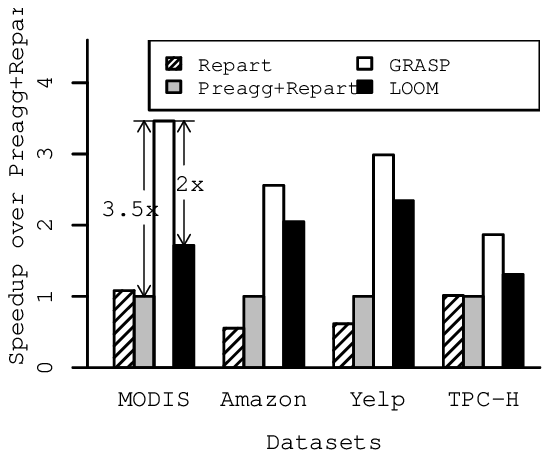}
\vspace{0.2em}
\caption{Speedup over Preagg+Repart on TPC-H workload and real datasets.}
\label{fig:exp:realdata}
\end{minipage}
\end{figure*}


\subsubsection{Impact of bandwidth estimation}
\label{sec:exp:bwest}

%

\noindent
The bandwidth estimation procedure described in Section~\ref{sec:alg:estbw}
leads to two questions: how accurate is the estimation and how robust is GRASP
to estimation errors?

Figure~\ref{fig:bwcomp} compares the available bandwidth as estimated by GRASP versus 
a manual calculation based on the hardware specifications, the network topology and the fragment placement.
This experiment uses 8 machines with each machine having 14 fragments in the experiment.
``Within machine'' and ``Across machines'' corresponds to the communication bandwidth between
fragments within the same node and across different nodes, respectively.
The result shows that the estimation error is within 20\% from the
theoretical bandwidth.
We conclude that the GRASP estimation procedure is fairly accurate in an idle cluster.

The estimation procedure may introduce errors in production clusters that are rarely idle. 
%
%
Figure~\ref{fig:bwest} shows the impact of bandwidth underestimation on
the response time of the aggregation plan produced by GRASP.
We test two underestimation levels, $20\%$ and $50\%$ from the theoretical value.
In this experiment we force GRASP to use a modified bandwidth matrix while
running the aggregation query on the MODIS dataset. 
We run the experiment 10 times picking nodes at random for each setting, and
show the standard deviation as an error bar.
Co-location results in the underestimation of the
communication bandwidth between local fragments in one or more machines.
NIC contention and switch contention underestimates the available network
bandwidth for one or all nodes in the cluster, respectively.
``Topology'' corresponds to the calculation based on the hardware capabilities,
while ``GRASP estimation'' corresponds to the procedure described in
Section~\ref{sec:alg:estbw}.
The horizontal axis is the response time difference with respect to the plan
GRASP generated using the theoretical hardware capabilities (hence, lower means faster).
The result shows that GRASP has better performance when using the estimated bandwidth
matrix than the accurate bandwidth from network topology.
This is because 
the estimated bandwidth measured from
the benchmark is closer to the available bandwidth during query execution.
Moreover, even when the available bandwidth is underestimated by up to 50\%,
the change in query response time is less than 20\%.
We conclude that GRASP is robust to errors introduced during bandwidth approximation.

\subsubsection{Effect of nonuniform bandwidth}
\label{sec:exp:nonuniformnetwork}

\noindent
GRASP takes network bandwidth into consideration in aggregation
scheduling.
How well does GRASP work when the bandwidth between network
links is different in a cluster?

In this experiment, we use 4 machines and each machine has 14 aggregation fragments.
The dataset in each fragment has 14 million tuples with $R.a$ ranging from 1 to 14 million.

The result is shown in Figure~\ref{fig:exp:topology}.
The vertical axis is the speedup over Preagg+Repart.
The results show that GRASP has better performance than both
repartitioning and LOOM in both all-to-one and all-to-all aggregations.
GRASP is up to $16\times$ faster than Preagg+Repart and $5.6\times$ faster
than LOOM in all-to-one aggregation and
$4.6\times$ faster than Preagg+Repart in all-to-all aggregation.
This is because GRASP is topology-aware and schedules more aggregations
on the faster network links.
GRASP is topology-aware and has better performance than the baselines when the bandwidth between fragments is not uniform.


\subsubsection{Effect of more plan fragments}
\label{sec:exp:scaleout}

\noindent
GRASP considers the candidate aggregations between all plan fragments for all
partitions in each phase of aggregation scheduling.
Hence the cost of GRASP increases when there are more plan fragments.
In this experiment, we evaluate how GRASP works when the number of fragments
increases.
We change the number of fragments from 28, 56, 84 to 112 by
running 14 fragments per node and changing the number of nodes from 2, 4, 6 to 8.
Each plan fragment has 16 million tuples with $R.a$ ranging from 1 to 16 million.

The result is shown in Figure~\ref{fig:exp:scaleout}, where the horizontal
axis is the number of fragments and the vertical axis is the speedup over Preagg+Repart.
For all-to-one aggregations, Figure~\ref{fig:exp:scaleout:alltoone} shows
that GRASP has better performance and
is $41\times$ faster than Preagg+Repart and $7.5\times$ faster than
LOOM when the number of fragments is 112.
The speedup increases when the number of fragments increases.
This is because in all-to-one aggregations the receiving link of the final
destination node is the bottleneck when repartitioning.
Hence, the performance of repartitioning rapidly degrades when the number of fragments increases.

For all-to-all aggregations, Figure~\ref{fig:exp:scaleout:alltoall} shows that GRASP
is $4.6\times$ faster than Preagg+Repart when the number of fragments is 56.
However, the speedup decreases for GRASP when the number of fragments exceeds 56 in
all-to-all aggregation.
This is because the planning 
cost of GRASP becomes more expensive in all-to-all aggregations as there
are more candidate transfers to consider in each phase.
This points to the need to parallelize aggregation planning for all-to-all
aggregations in networks that concurrently execute hundreds of plan fragments.


\subsubsection{Real datasets and the TPC-H workload}
\label{sec:exp:realdata}

\noindent
These experiments evaluate the performance of the GRASP plans 
with the TPC-H workload and three real datasets.
We use 8 machines and 14 fragments per machine.
The dataset is aggregated to fragment 0, which corresponds to
the all-to-one aggregation.

\vspace{0.75em}
\noindent
\textbf{Speedup results:}
Figure~\ref{fig:exp:realdata} shows the 
speedup over Preagg
+Repart for each algorithm.
The result shows that GRASP has the best performance for all datasets.
GRASP is $2\times$ faster than LOOM and $3.5\times$ faster than Preagg+Repart 
in the MODIS dataset.


\vspace{0.75em}
\noindent
\textbf{Network utilization:}
Figure~\ref{fig:netutil} shows the network utilization plot for the MODIS dataset.
The horizontal axis is the time elapsed since the query was submitted to the coordinator.
(Note that the scale of the horizontal axis is not the same, as some
algorithms finish earlier than others.)
Each horizontal line in the plot represents one incoming network link or one outgoing link
of a fragment.
For each link, we plot a line when there is traffic in the link and leave it blank
otherwise.

Figure~\ref{fig:netutil:grasp} shows network utilization with GRASP.
After a short delay to compute the aggregation plan, 
the network is fully utilized in the first few
phases and there is traffic in all links.
As the aggregation progresses, more fragments contain no data and hence
these fragments do not further participate in the aggregation.
The aggregation finishes in under 300 seconds.

Figure~\ref{fig:netutil:loom} shows LOOM. One can see that the network
links, especially the receiving links, are not as fully utilized as in
Figure~\ref{fig:netutil:grasp}.
The fan-in of the aggregation tree produced by LOOM is 5 for this
experiment, which makes the receiving link of the parent fragment to be
bottleneck.
The aggregation finishes in about 600 seconds.

Figure~\ref{fig:netutil:preagg} shows Preagg+Repart.
All receiving links except fragment 0 (the aggregation destination) are
not utilized. The entire aggregation is bottlenecked on the receiving
capability of fragment 0.
The aggregation takes more than 900 seconds.
We omit the figure for Repart as it is similar to Preagg+Repart.

\begin{figure}[]
\centering
\begin{subfigure}{0.15\textwidth}
\centering
\includegraphics[scale=0.7]{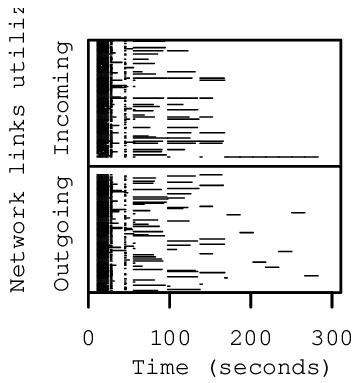}
\caption{GRASP}
\label{fig:netutil:grasp}
\end{subfigure}
\begin{subfigure}{0.15\textwidth}
\centering
\includegraphics[scale=0.7]{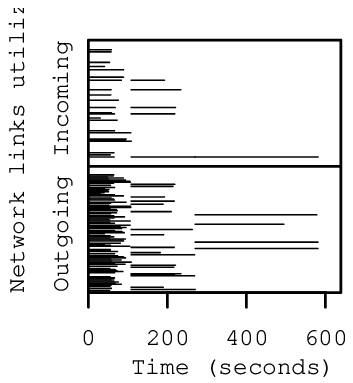}
\caption{LOOM}
\label{fig:netutil:loom}
\end{subfigure}
\begin{subfigure}{0.15\textwidth}
\centering
\includegraphics[scale=0.7]{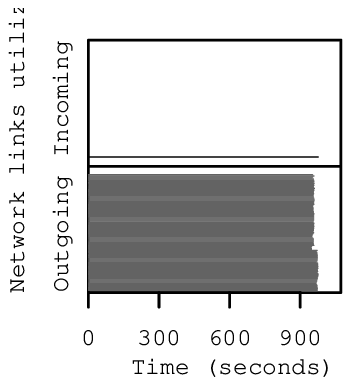}
\caption{Preagg+repart}
\label{fig:netutil:preagg}
\end{subfigure}
\vspace{0.25em}
\caption{Network link utilization.}
\vspace{-0.5em}
\label{fig:netutil}
\end{figure}

\vspace{0.5em}
\noindent
\textbf{Tuples transmitted to destination:}
The GRASP performance gains can be directly attributed to the fact that
it transmits less data on the incoming link of the destination fragment,
which is frequently the bottleneck of the entire aggregation.
Table~\ref{tab:exp:modisri2:f0tup} shows
how many tuples the destination fragment receives under different algorithms.
Local pre-aggregation has minimal impact as it is only effective when duplicate
keys happen to be co-located on the same node.
LOOM transmits fewer tuples to the destination fragment as tuples are
combined in the aggregation tree before arriving at the final
destination fragment.
By aggressively combining fragments based on their similarity, GRASP
transmits $2.7\times$ less tuples than LOOM to the destination fragment.

\vspace{0.5em}
\noindent
\textbf{Accuracy of minhash estimation:}
We also evaluate the accuracy of the minhash estimation with the MODIS dataset.
Figure~\ref{fig:exp:minhash} shows the cumulative distribution function
of the absolute error in estimating the size of the intersection between
fragments when the cardinality of the input is accurately known.
The result shows that 
the absolute error of the size of the intersection is less than $10\%$
for $90\%$ of the estimations. 
We conclude that the minhash estimation is accurate and it allows
GRASP to pick suitable fragment pairs for aggregation.%



\begin{table}[t]
\caption{Tuples received by the final destination fragment.}
\label{tab:exp:modisri2:f0tup}
\small
\begin{tabular}{cccc}
\toprule
Repart     & Preagg+Repart & LOOM & GRASP      \\
\midrule
3,464,926,620 & 3,195,388,849 & 2,138,236,114  & 787,105,152 \\
\bottomrule
\end{tabular}
\vspace{-2em}
\end{table}

\begin{figure}[t]
\begin{minipage}{0.43\linewidth}
\includegraphics[scale=0.8]{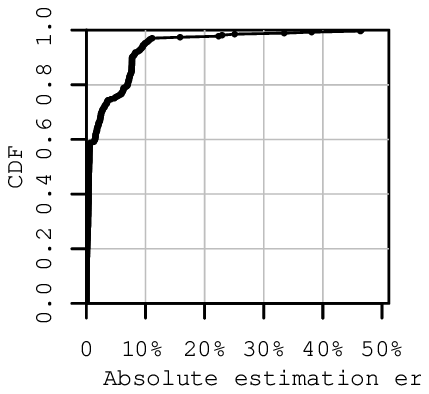}
\caption{Absolute error in minhash estimation.}
\vspace{-1em}
\label{fig:exp:minhash}
\end{minipage}
\hspace{1em}
\begin{minipage}{0.52\linewidth}
\includegraphics[scale=0.8]{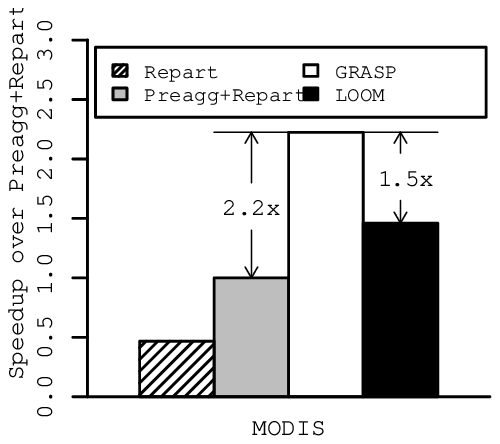}
\caption{Speedup over Preagg+Repart on the MODIS dataset in Amazon EC2.}
\vspace{-1em}
\label{fig:modisec2}
\end{minipage}
\end{figure}

\subsubsection{Evaluation on Amazon EC2}
\label{sec:exp:ec2}


\noindent
This section evaluates GRASP on the MODIS dataset on Amazon EC2.
We allocate 8 instances of type d2.8xlarge and run 6 fragments in each instance.
Figure~\ref{fig:modisec2} shows the 
speedup over the Preagg+Repart algorithm for each algorithm.
Preagg+Repart has better performance than Repart in this experiment.
This is because the fast 10 Gbps network in EC2 makes the query compute bound. 
The throughput of the local aggregation on pre-aggregated data is measured to be 811 MB/s,
which is faster than aggregation on raw data with throughput to be 309 MB/s.
This does not make a difference in the experiment in Section~\ref{sec:exp:realdata}, as
aggregation is network bound in the 1 Gbps network where the maximum throughput is 125 MB/s.
However, the aggregation is compute bound in the 10 Gbps network of EC2 with a
maximum throughput of 1.2 GB/s, hence pre-aggregation makes a big difference.

Figure~\ref{fig:modisec2} shows that GRASP is 2.2$\times$ faster than
Preagg+
Repart and 1.5$\times$ faster than LOOM.
GRASP still has better performance when computation is the bottleneck.
This is because GRASP maximizes network utilization by scheduling as many
aggregations as possible in each phase, which also maximizes the number of
fragments participating in the aggregation and sharing the computation load of each phase.%


%

\section{Related work}
\label{sec:relwork}
\subsubsection*{Aggregation execution}
\noindent
Aggregation has been extensively studied in previous works.
Many works have focused on how to execute an aggregation efficiently in a single server.
Larson~\cite{larson:agg02icde} studied how to use partial aggregation
to reduce the input size of other operations.
Cieslewicz and Ross~\cite{Cieslewicz:2007vldb} evaluated aggregation algorithms 
with independent and shared hash tables on multi-core processors.
Ye et al.~\cite{ye:agg11damon} compared different in-memory parallel aggregation algorithms on the Intel Nehalem architecture.
Raman et al.~\cite{Raman:2013vldb} described the grouping and aggregation algorithm used in DB2 BLU.
M\"{u}ller et al.~\cite{Muller:2015sigmod} proposed an
adaptive algorithm which combines the hashing and sorting implementations.
Wang et al.~\cite{liwang:numaagg15} proposed a NUMA-aware aggregation algorithm.
Jiang and Gagan~\cite{gagan:simd17} and Polychroniou et al~\cite{Polychroniou:2015simd} 
used SIMD and MIMD to parallelize the execution of aggregation.
Gan et al.~\cite{momentagg} optimized high cardinality aggregation queries with moment based summaries.
M{\"{u}}ller et al.~\cite{gustavo18agg} studied the floating-point aggregation.

Aggregation has also been studied in the parallel database system literature.
Graefe~\cite{graefe:93} introduced aggregation evaluation techniques in parallel database system.
Shatdal and Naughton~\cite{Shatdal:parallelagg95sigmod} proposed adaptive
algorithms which switch between the repartition and the two-phase algorithm at runtime.
Aggregation trees are used in accelerating parallel aggregations.
Melnik et al.~\cite{dremel10vldb} introduced Dremel, which uses a multi-level
serving tree to execute aggregation queries.
Yuan et al.~\cite{agginterface:sosp09} compared the interfaces and implementations for user-defined
distributed aggregations in several distributed computing systems.
Mai et al.~\cite{Mai14Netagg} implemented NetAgg which aggregates data along network paths.
Costa et al.~\cite{camdoop12nsdi} proposed Camdoop, which does in-network aggregation for a MapReduce-like system
in a cluster with a direct-connect network topology.
Yalagandula and Dahlin~\cite{sdims04sigcomm} designed a distributed information management system to do hierarchical aggregation in networked systems.
Culhane et al.~\cite{Culhane2014hotcloud,Culhane2015infocomm} proposed LOOM, which builds an aggregation
tree with fixed fan-in for all-to-one aggregations.

The impact of the network topology on aggregation has been studied.
Gupta et al.~\cite{gupta:dsn01} proposed an aggregation algorithm that works in unreliable networks
such as sensor networks.
Madden et al.~\cite{Madden:2003sigmod} designed an acquisitional query processor for sensor networks
to reduce power in query evaluation.
Madden et al.~\cite{madden02osdi, madden02mcsa} also proposed a tiny aggregation service which does
in network aggregation in sensor networks.
Chowdhury et al.~\cite{orchestra11sigcomm} proposed Orchestra to manage network activities
in MapReduce systems.

None of the above aggregation algorithms takes advantage of the similarity between fragments as
GRASP does.
The most relevant work is LOOM which considers the amount of data reduction in an aggregation during planning.
However LOOM only considers the overall reduction rate and does not consider data similarities during aggregation.
The biggest strength of GRASP is that it carefully estimates the size of every
partial aggregation and handles each partition differently, which
is not possible with LOOM. 

\subsubsection*{Distribution-aware algorithms}
\noindent
Distribution-aware algorithms use information about the distribution and the placement of the data during query processing.
Prior works have extensively studied how to take advantage of locality.
Some algorithms consider the offline setting.
Zamanian et al.~\cite{zamanian:partition15sigmod} introduced a data partitioning algorithm to maximize locality in the data distribution.
Prior works have also considered how to extract and exploit locality information at runtime.
R{\"{o}}diger et al.~\cite{Rodiger:14icde} proposed a locality-sensitive join algorithm which first
builds a histogram for the workload, then schedules the join execution to reduce network traffic.
Polychroniou~\cite{trackjoin:sigmod15} proposed track-join, where the distribution of the join key is exchanged
across the cluster to generate a join schedule to leverage locality.
Lu et al.~\cite{AdaptDB} proposed AdaptDB, which refines data partitioning according to access patterns
at runtime.

Distribution-aware algorithms have also been proposed to deal with skewed datasets.
DeWitt et al.~\cite{DeWitt:joinskew} handled skew in a join by first sampling the data, then partitioning the build relation and replicating the probe relation as needed.
Shah et al.~\cite{flux03} implemented an adaptive partitioning operator to 
collect dataset information at runtime and address the problem of workload imbalance in continuous query systems.
Xu et al.~\cite{xu:sigmod08} addressed skew in parallel joins by 
first scanning the dataset to identify the skewed values, then keeping the
skewed rows locally and duplicating the matching rows.
R{\"{o}}diger et al.~\cite{flowjoin16} adopted similar approach as DeWitt et al.~\cite{DeWitt:joinskew} by first sampling 1\% of the data and then use this information to decide the data partition scheme.
Wolf et al.~\cite{wolfjoinskew:93} divided the parallel hash join into two phases, and add one scheduling phase to split the partition with data skew.
Elseidy et al.~\cite{Elseidy:2014} proposed a parallel online dataflow join which is resilient to data skew.


\vspace{-0.5em}
\section{Conclusions and future work}
\label{sec:conclude}




\noindent
Parallel aggregation is a ubiquitous operation in data analytics.
For low-cardinality parallel aggregations, the network cost is
negligible after the data has been aggregated locally using pre-aggregation.
However, the network communication cost becomes
significant for high-cardinality parallel aggregations.
This paper proposes GRASP, an algorithm that schedules parallel
aggregation in a distribution-aware manner to increase network
utilization and reduce the communication cost for algebraic
aggregations.

Looking ahead, GRASP can be further extended in two promising ways. 
First, GRASP can be extended for non-algebraic aggregations.
This would require a new metric to quantify the data reduction of an 
aggregation pair. 
Second, the assumption that the communication cost dominates the
aggregation marginally holds on 10 Gbps networks, and will not hold in
faster networks such as 
InfiniBand. One opportunity is to 
augment the cost estimation formulas to account for compute overheads,
instead of modeling the network transfer cost alone.
This can jointly optimize compute and
communication overheads during aggregation in high-performance networks.

\vspace{0.5em}
\noindent
\textbf{Acknowledgements:}
We would like to acknowledge Srinivasan Parthasarathy, Jiongqian Liang,
Vishal Dey and the anonymous reviewers for their
insightful comments that improved this paper.
This work was supported by the National Science Foundation
grants IIS-1464381, CCF-1816577, CCF-1815145, CCF-1423230 and CAREER award 1453472.

\balance
\bibliographystyle{abbrv}
\bibliography{ms}

\begin{thebibliography}{10}

\bibitem{fattree}
M.~Al-Fares, A.~Loukissas, and A.~Vahdat.
\newblock {A Scalable, Commodity Data Center Network Architecture}.
\newblock {\em SIGCOMM Comput. Commun. Rev.}, 38(4):63--74, Aug. 2008.

\bibitem{austrin2012inapproximability}
P.~Austrin, T.~Pitassi, and Y.~Wu.
\newblock {Inapproximability of Treewidth, One-shot Pebbling, and Related
  Layout Problems}.
\newblock In {\em Approximation, Randomization, and Combinatorial Optimization.
  Algorithms and Techniques}, pages 13--24. Springer, 2012.

\bibitem{jaccard98}
A.~Z. Broder, M.~Charikar, A.~M. Frieze, and M.~Mitzenmacher.
\newblock {Min-wise Independent Permutations (Extended Abstract)}.
\newblock In {\em Proceedings of the Thirtieth Annual ACM Symposium on Theory
  of Computing}, STOC '98, pages 327--336, New York, NY, USA, 1998. ACM.

\bibitem{cayley1889}
A.~Cayley.
\newblock A theorem on trees.
\newblock {\em Quarterly Journal of Pure Applied Mathematics}, 23:376--378,
  1889.

\bibitem{orchestra11sigcomm}
M.~Chowdhury, M.~Zaharia, J.~Ma, M.~I. Jordan, and I.~Stoica.
\newblock {Managing Data Transfers in Computer Clusters with Orchestra}.
\newblock In {\em Proceedings of the ACM SIGCOMM 2011 Conference}, SIGCOMM '11,
  pages 98--109, New York, NY, USA, 2011. ACM.

\bibitem{Cieslewicz:2007vldb}
J.~Cieslewicz and K.~A. Ross.
\newblock {Adaptive Aggregation on Chip Multiprocessors}.
\newblock In {\em Proceedings of the 33rd International Conference on Very
  Large Data Bases}, VLDB '07, pages 339--350. VLDB Endowment, 2007.

\bibitem{camdoop12nsdi}
P.~Costa, A.~Donnelly, A.~I.~T. Rowstron, and G.~O'Shea.
\newblock {Camdoop: Exploiting In-network Aggregation for Big Data
  Applications}.
\newblock In {\em Proceedings of the 9th {USENIX} Symposium on Networked
  Systems Design and Implementation, {NSDI} 2012, San Jose, CA, USA, April
  25-27, 2012}, pages 29--42, 2012.

\bibitem{Culhane2014hotcloud}
W.~Culhane, K.~Kogan, C.~Jayalath, and P.~Eugster.
\newblock {LOOM: Optimal Aggregation Overlays for In-memory Big Data
  Processing}.
\newblock In {\em Proceedings of the 6th USENIX Conference on Hot Topics in
  Cloud Computing}, HotCloud'14, pages 13--13, Berkeley, CA, USA, 2014. USENIX
  Association.

\bibitem{Culhane2015infocomm}
W.~Culhane, K.~Kogan, C.~Jayalath, and P.~Eugster.
\newblock {Optimal communication structures for big data aggregation}.
\newblock In {\em 2015 {IEEE} Conference on Computer Communications, {INFOCOM}
  2015, Kowloon, Hong Kong, April 26 - May 1, 2015}, pages 1643--1651, 2015.

\bibitem{DeWitt:joinskew}
D.~J. DeWitt, J.~F. Naughton, D.~A. Schneider, and S.~Seshadri.
\newblock {Practical Skew Handling in Parallel Joins}.
\newblock In {\em Proceedings of the 18th International Conference on Very
  Large Data Bases}, VLDB '92, pages 27--40, San Francisco, CA, USA, 1992.
  Morgan Kaufmann Publishers Inc.

\bibitem{Elseidy:2014}
M.~Elseidy, A.~Elguindy, A.~Vitorovic, and C.~Koch.
\newblock {Scalable and Adaptive Online Joins}.
\newblock {\em {PVLDB}}, 7(6):441--452, 2014.

\bibitem{momentagg}
E.~Gan, J.~Ding, K.~S. Tai, V.~Sharan, and P.~Bailis.
\newblock {Moment-Based Quantile Sketches for Efficient High Cardinality
  Aggregation Queries}.
\newblock {\em CoRR}, abs/1803.01969, 2018.

\bibitem{lsh99}
A.~Gionis, P.~Indyk, and R.~Motwani.
\newblock {Similarity Search in High Dimensions via Hashing}.
\newblock In {\em Proceedings of the 25th International Conference on Very
  Large Data Bases}, VLDB '99, pages 518--529, San Francisco, CA, USA, 1999.
  Morgan Kaufmann Publishers Inc.

\bibitem{graefe:93}
G.~Graefe.
\newblock {Query Evaluation Techniques for Large Databases}.
\newblock {\em {ACM} Comput. Surv.}, 25(2):73--170, 1993.

\bibitem{graspcode}
{GRASP}.
\newblock {\url{https://code.osu.edu/pythia/grasp}}.

\bibitem{datacube96icde}
J.~Gray, A.~Bosworth, A.~Layman, and H.~Pirahesh.
\newblock {Data Cube: {A} Relational Aggregation Operator Generalizing
  Group-By, Cross-Tab, and Sub-Total}.
\newblock In {\em Proceedings of the Twelfth International Conference on Data
  Engineering, February 26 - March 1, 1996, New Orleans, Louisiana}, pages
  152--159, 1996.

\bibitem{datacenter}
A.~G. Greenberg, J.~R. Hamilton, N.~Jain, S.~Kandula, C.~Kim, P.~Lahiri, D.~A.
  Maltz, P.~Patel, and S.~Sengupta.
\newblock {VL2: A Scalable and Flexible Data Center Network}.
\newblock In {\em Proceedings of the {ACM} {SIGCOMM} 2009 Conference on
  Applications, Technologies, Architectures, and Protocols for Computer
  Communications, Barcelona, Spain, August 16-21, 2009}, pages 51--62, 2009.

\bibitem{gupta:dsn01}
I.~Gupta, R.~v. Renesse, and K.~P. Birman.
\newblock {Scalable Fault-Tolerant Aggregation in Large Process Groups}.
\newblock In {\em Proceedings of the 2001 International Conference on
  Dependable Systems and Networks (Formerly: FTCS)}, DSN '01, pages 433--442,
  Washington, DC, USA, 2001. IEEE Computer Society.

\bibitem{amazondataset}
R.~He and J.~McAuley.
\newblock {Ups and Downs: Modeling the Visual Evolution of Fashion Trends with
  One-Class Collaborative Filtering}.
\newblock In {\em Proceedings of the 25th International Conference on World
  Wide Web, {WWW} 2016, Montreal, Canada, April 11 - 15, 2016}, pages 507--517,
  2016.

\bibitem{hpctopo}
{\url{https://htor.inf.ethz.ch/research/topologies/}}.

\bibitem{lsh98}
P.~Indyk and R.~Motwani.
\newblock {Approximate Nearest Neighbors: Towards Removing the Curse of
  Dimensionality}.
\newblock In {\em Proceedings of the Thirtieth Annual ACM Symposium on Theory
  of Computing}, STOC '98, pages 604--613, New York, NY, USA, 1998. ACM.

\bibitem{gagan:simd17}
P.~Jiang and G.~Agrawal.
\newblock {Efficient SIMD and MIMD Parallelization of Hash-based Aggregation by
  Conflict Mitigation}.
\newblock In {\em Proceedings of the International Conference on
  Supercomputing}, ICS '17, pages 24:1--24:11, New York, NY, USA, 2017. ACM.

\bibitem{larson:agg02icde}
P.~Larson.
\newblock {Data Reduction by Partial Preaggregation}.
\newblock In {\em Proceedings of the 18th International Conference on Data
  Engineering, San Jose, CA, USA, February 26 - March 1, 2002}, pages 706--715,
  2002.

\bibitem{leis18vldb}
V.~Leis, B.~Radke, A.~Gubichev, A.~Mirchev, P.~A. Boncz, A.~Kemper, and
  T.~Neumann.
\newblock {Query Optimization Through the Looking Glass, and What We Found
  Running the Join Order Benchmark}.
\newblock {\em {VLDB} J.}, 27(5):643--668, 2018.

\bibitem{grasp}
F.~Liu, A.~Salmasi, S.~Blanas, and A.~Sidiropoulos.
\newblock Chasing similarity: Distribution-aware aggregation scheduling.
\newblock {\em {PVLDB}}, 12(3):292--306, 2018.

\bibitem{AdaptDB}
Y.~Lu, A.~Shanbhag, A.~Jindal, and S.~Madden.
\newblock {{AdaptDB: A}daptive Partitioning for Distributed Joins}.
\newblock {\em {PVLDB}}, 10(5):589--600, 2017.

\bibitem{luo18socc}
L.~Luo, J.~Nelson, L.~Ceze, A.~Phanishayee, and A.~Krishnamurthy.
\newblock {Parameter Hub: a Rack-Scale Parameter Server for Distributed Deep
  Neural Network Training}.
\newblock In {\em Proceedings of the {ACM} Symposium on Cloud Computing, SoCC
  2018, Carlsbad, CA, USA, October 11-13, 2018}, pages 41--54, 2018.

\bibitem{madden02osdi}
S.~Madden, M.~J. Franklin, J.~M. Hellerstein, and W.~Hong.
\newblock {{TAG:} {A} Tiny AGgregation Service for Ad-Hoc Sensor Networks}.
\newblock In {\em 5th Symposium on Operating System Design and Implementation
  {(OSDI} 2002), Boston, Massachusetts, USA, December 9-11, 2002}, 2002.

\bibitem{Madden:2003sigmod}
S.~Madden, M.~J. Franklin, J.~M. Hellerstein, and W.~Hong.
\newblock {The Design of an Acquisitional Query Processor for Sensor Networks}.
\newblock In {\em Proceedings of the 2003 ACM SIGMOD International Conference
  on Management of Data}, SIGMOD '03, pages 491--502, New York, NY, USA, 2003.
  ACM.

\bibitem{madden02mcsa}
S.~Madden, R.~Szewczyk, M.~J. Franklin, and D.~E. Culler.
\newblock {Supporting Aggregate Queries Over Ad-Hoc Wireless Sensor Networks}.
\newblock In {\em 4th {IEEE} Workshop on Mobile Computing Systems and
  Applications {(WMCSA} 2002), 20-21 June 2002, Callicoon, NY, {USA}}, pages
  49--58, 2002.

\bibitem{Mai14Netagg}
L.~Mai, L.~Rupprecht, A.~Alim, P.~Costa, M.~Migliavacca, P.~Pietzuch, and A.~L.
  Wolf.
\newblock {NetAgg: Using Middleboxes for Application-specific On-path
  Aggregation in Data Centres}.
\newblock In {\em Proceedings of the 10th ACM International on Conference on
  Emerging Networking Experiments and Technologies}, CoNEXT '14, pages
  249--262, New York, NY, USA, 2014. ACM.

\bibitem{dremel10vldb}
S.~Melnik, A.~Gubarev, J.~J. Long, G.~Romer, S.~Shivakumar, M.~Tolton, and
  T.~Vassilakis.
\newblock {Dremel: Interactive Analysis of Web-Scale Datasets}.
\newblock {\em {PVLDB}}, 3(1):330--339, 2010.

\bibitem{gustavo18agg}
I.~M{\"{u}}ller, A.~Arteaga, T.~Hoefler, and G.~Alonso.
\newblock {Reproducible Floating-Point Aggregation in RDBMSs}.
\newblock {\em CoRR}, abs/1802.09883, 2018.

\bibitem{Muller:2015sigmod}
I.~M\"{u}ller, P.~Sanders, A.~Lacurie, W.~Lehner, and F.~F\"{a}rber.
\newblock {Cache-Efficient Aggregation: Hashing Is Sorting}.
\newblock In {\em Proceedings of the 2015 ACM SIGMOD International Conference
  on Management of Data}, SIGMOD '15, pages 1123--1136, New York, NY, USA,
  2015. ACM.

\bibitem{Polychroniou:2015simd}
O.~Polychroniou, A.~Raghavan, and K.~A. Ross.
\newblock {Rethinking SIMD Vectorization for In-Memory Databases}.
\newblock In {\em Proceedings of the 2015 ACM SIGMOD International Conference
  on Management of Data}, SIGMOD '15, pages 1493--1508, New York, NY, USA,
  2015. ACM.

\bibitem{trackjoin:sigmod15}
O.~Polychroniou, R.~Sen, and K.~A. Ross.
\newblock {Track Join: Distributed Joins with Minimal Network Traffic}.
\newblock In {\em Proceedings of the 2014 ACM SIGMOD International Conference
  on Management of Data}, SIGMOD '14, pages 1483--1494, New York, NY, USA,
  2014. ACM.

\bibitem{raghavendra2010graph}
P.~Raghavendra and D.~Steurer.
\newblock {Graph Expansion and the Unique Games Conjecture}.
\newblock In {\em Proceedings of the forty-second ACM symposium on Theory of
  computing}, pages 755--764. ACM, 2010.

\bibitem{raghavendra2012reductions}
P.~Raghavendra, D.~Steurer, and M.~Tulsiani.
\newblock {Reductions Between Expansion Problems}.
\newblock In {\em Computational Complexity (CCC), 2012 IEEE 27th Annual
  Conference on}, pages 64--73. IEEE, 2012.

\bibitem{Raman:2013vldb}
V.~Raman, G.~Attaluri, R.~Barber, N.~Chainani, D.~Kalmuk, V.~KulandaiSamy,
  J.~Leenstra, S.~Lightstone, S.~Liu, G.~M. Lohman, T.~Malkemus, R.~Mueller,
  I.~Pandis, B.~Schiefer, D.~Sharpe, R.~Sidle, A.~Storm, and L.~Zhang.
\newblock {DB2 with BLU Acceleration: So Much More Than Just a Column Store}.
\newblock {\em {PVLDB}}, 6(11):1080--1091, 2013.

\bibitem{flowjoin16}
W.~R{\"{o}}diger, S.~Idicula, A.~Kemper, and T.~Neumann.
\newblock {Flow-Join: Adaptive Skew Handling for Distributed Joins over
  High-speed Networks}.
\newblock In {\em 32nd {IEEE} International Conference on Data Engineering,
  {ICDE} 2016, Helsinki, Finland, May 16-20, 2016}, pages 1194--1205, 2016.

\bibitem{wolf15vldb}
W.~R{\"{o}}diger, T.~M{\"{u}}hlbauer, A.~Kemper, and T.~Neumann.
\newblock {High-Speed Query Processing over High-Speed Networks}.
\newblock {\em {PVLDB}}, 9(4):228--239, 2015.

\bibitem{Rodiger:14icde}
W.~R{\"{o}}diger, T.~M{\"{u}}hlbauer, P.~Unterbrunner, A.~Reiser, A.~Kemper,
  and T.~Neumann.
\newblock {Locality-sensitive Operators for Parallel Main-memory Database
  Clusters}.
\newblock In {\em {IEEE} 30th International Conference on Data Engineering,
  Chicago, {ICDE} 2014, IL, USA, March 31 - April 4, 2014}, pages 592--603,
  2014.

\bibitem{Satuluri:2012}
V.~Satuluri and S.~Parthasarathy.
\newblock {Bayesian Locality Sensitive Hashing for Fast Similarity Search}.
\newblock {\em {PVLDB}}, 5(5):430--441, 2012.

\bibitem{flux03}
M.~A. Shah, J.~M. Hellerstein, S.~Chandrasekaran, and M.~J. Franklin.
\newblock {Flux: An Adaptive Partitioning Operator for Continuous Query
  Systems}.
\newblock In {\em Proceedings of the 19th International Conference on Data
  Engineering, March 5-8, 2003, Bangalore, India}, pages 25--36, 2003.

\bibitem{Shatdal:parallelagg95sigmod}
A.~Shatdal and J.~F. Naughton.
\newblock {Adaptive Parallel Aggregation Algorithms}.
\newblock In {\em Proceedings of the 1995 ACM SIGMOD International Conference
  on Management of Data}, SIGMOD '95, pages 104--114, New York, NY, USA, 1995.
  ACM.

\bibitem{modis09}
E.~{\relax Vermote-NASA GSFC and MODAPS SIPS - NASA. (2015)}.
\newblock {MOD09 MODIS/Terra L2 Surface Reflectance, 5-Min Swath 250m, 500m,
  and 1km}.
\newblock NASA LP DAAC.

\bibitem{modis09A1}
E.~{\relax Vermote-NASA GSFC and MODAPS SIPS - NASA. (2015)}.
\newblock {MOD09A1 MODIS/Surface Reflectance 8-Day L3 Global 500m SIN Grid}.
\newblock NASA LP DAAC.

\bibitem{liwang:numaagg15}
L.~Wang, M.~Zhou, Z.~Zhang, M.~Shan, and A.~Zhou.
\newblock {NUMA-Aware Scalable and Efficient In-Memory Aggregation on Large
  Domains}.
\newblock {\em {IEEE} Trans. Knowl. Data Eng.}, 27(4):1071--1084, 2015.

\bibitem{wolfjoinskew:93}
J.~L. Wolf, P.~S. Yu, J.~Turek, and D.~M. Dias.
\newblock {A Parallel Hash Join Algorithm for Managing Data Skew}.
\newblock {\em IEEE Trans. Parallel Distrib. Syst.}, 4(12):1355--1371, Dec.
  1993.

\bibitem{xu:sigmod08}
Y.~Xu, P.~Kostamaa, X.~Zhou, and L.~Chen.
\newblock {Handling Data Skew in Parallel Joins in Shared-nothing Systems}.
\newblock In {\em Proceedings of the 2008 ACM SIGMOD International Conference
  on Management of Data}, SIGMOD '08, pages 1043--1052, New York, NY, USA,
  2008. ACM.

\bibitem{sdims04sigcomm}
P.~Yalagandula and M.~Dahlin.
\newblock {A Scalable Distributed Information Management System}.
\newblock In {\em Proceedings of the 2004 Conference on Applications,
  Technologies, Architectures, and Protocols for Computer Communications},
  SIGCOMM '04, pages 379--390, New York, NY, USA, 2004. ACM.

\bibitem{ye:agg11damon}
Y.~Ye, K.~A. Ross, and N.~Vesdapunt.
\newblock {Scalable Aggregation on Multicore Processors}.
\newblock In {\em Proceedings of the Seventh International Workshop on Data
  Management on New Hardware}, DaMoN '11, pages 1--9, New York, NY, USA, 2011.
  ACM.

\bibitem{yelpdataset}
\url{https://www.yelp.com/dataset/documentation/json}.

\bibitem{agginterface:sosp09}
Y.~Yu, P.~K. Gunda, and M.~Isard.
\newblock {Distributed Aggregation for Data-parallel Computing: Interfaces and
  Implementations}.
\newblock In {\em Proceedings of the ACM SIGOPS 22Nd Symposium on Operating
  Systems Principles}, SOSP '09, pages 247--260, New York, NY, USA, 2009. ACM.

\bibitem{zamanian:partition15sigmod}
E.~Zamanian, C.~Binnig, and A.~Salama.
\newblock {Locality-aware Partitioning in Parallel Database Systems}.
\newblock In {\em Proceedings of the 2015 ACM SIGMOD International Conference
  on Management of Data}, SIGMOD '15, pages 17--30, New York, NY, USA, 2015.
  ACM.

\end{thebibliography}

\end{document}